%% file: main.tex
\documentclass[12pt]{article}

\usepackage[T1]{fontenc}
\usepackage[utf8]{inputenc}
\usepackage[a4paper,left=1in,right=1in,top=1.5in,bottom=1.5in]{geometry}
\usepackage{setspace}
\usepackage{microtype}
\usepackage{lmodern}

\usepackage{amsmath,amssymb,amsfonts,amsthm,bm,mathtools}

\usepackage{graphicx}
\graphicspath{{figures/}}
\usepackage{booktabs}
\usepackage{tabularx}
\usepackage{array}
\usepackage{caption}
\usepackage{subcaption}
\usepackage{threeparttable}
\usepackage{adjustbox}
\usepackage{float}
\usepackage[section]{placeins}
\usepackage{xcolor}
\usepackage{longtable}
\usepackage{pdflscape}
\newcolumntype{L}[1]{>{\raggedright\arraybackslash}p{#1}}
\newcolumntype{C}[1]{>{\centering\arraybackslash}p{#1}}

\usepackage[round,authoryear]{natbib}
\usepackage[
    colorlinks=true,
    allcolors=blue,
    bookmarks=false,
    pdftitle={The Price of Permission: Classification Uncertainty in Constrained Capital Markets},
    pdfauthor={Abdulrahman Qadi, Akash Sharma, and Francesca Medda}
]{hyperref}

\newcommand{\E}{\mathbb{E}}

\DeclareMathOperator{\Prb}{Pr}

\newcommand{\Disc}{\mathrm{DISC}}

\newcommand{\CAR}{\mathrm{CAR}}
\newcommand{\diffCAR}{\mathrm{diffCAR}}
\newcommand{\AltText}[1]{}
\theoremstyle{plain}
\newtheorem{proposition}{Proposition}
\newtheorem{corollary}{Corollary}
\theoremstyle{definition}
\newtheorem{assumption}{Assumption}

\newtheorem{remark}{Remark}

\title{\textbf{The Price of Permission: Classification Uncertainty in Constrained Capital Markets}}
\author{Abdulrahman Qadi, Akash Sharma, and Francesca Medda}
\date{}

\begin{document}
\pagenumbering{arabic}
\setcounter{page}{1}

\begin{center}
{\LARGE\bfseries The Price of Permission: Classification Uncertainty in Constrained Capital Markets\par}
\vspace{1em}
{\large Abdulrahman Qadi\textsuperscript{*}\quad Akash Sharma\quad Francesca Medda\par}
\vspace{0.55em}
{\small Institute of Finance and Technology, University College London\\[-0.1em]
Gower Street, London WC1E 6BT, United Kingdom\par}
\vspace{0.35em}
{\small \textsuperscript{*}Corresponding author: \href{mailto:abdulrahman.qadi.24@ucl.ac.uk}{abdulrahman.qadi.24@ucl.ac.uk}\par}
\end{center}

\vspace{0.8em}

\begin{abstract}
	\noindent
		Shariah-compliant equity screening provides a transparent setting in which institutional rules determine who may own a stock. A binary label identifies current eligibility but not whether the feasible investor base is fragmented across standards or close to changing. We define this instability as classification uncertainty and formalize its investor-base consequence through permitted investor mass. In a 1999--2024 CRSP--Compustat panel of 13,188 securities classified under seven researcher-emulated Shariah rulebooks, screening-rule disagreement and proximity to active boundaries rank next-month screen-implied transitions. U.S. Fama--MacBeth diagnostics do not support an unconditional equal-weighted permission premium, and a September 2023 DJIM/S\&P methodology change produces no robust matched repricing. The central event evidence uses 25 official Securities Commission Malaysia lists. The 410 inclusions already trading before the preceding review have positive but imprecise matched returns. Applying the pre-event turnover floor yields 295 inclusions with 1.76 percentage points over $[0,10]$ trading days ($p_{\mathrm{date}}=0.008$; $p_{\mathrm{wild}}=0.017$) and 2.25 points over $[0,20]$ ($p_{\mathrm{date}}=0.018$; $p_{\mathrm{wild}}=0.035$). Leave-one-date-out, first-inclusion-only, and mid-review placebo checks are supportive, although a joint 20-day pre-event test rejects. Ownership and demand-pressure diagnostics do not identify a unique marginal buyer or clean causal demand shock. The evidence supports treating classification risk as a portfolio-monitoring state. Official Shariah permission is associated with price effects in a recognized local market among sufficiently tradable securities; formal eligibility alone is insufficient.
	
\end{abstract}



\clearpage

\section{Introduction}\label{sec:intro}

Shariah-compliant equity screening provides an unusually transparent setting for a mainstream finance question: who is allowed to own a stock? A security enters a constrained mandate's feasible set only if it passes the financial and business-activity screens recognized by that mandate. Yet there is no single operational Shariah equity rulebook. AAOIFI, DJIM, S\&P, FTSE/Yasaar, MSCI, and Securities Commission Malaysia (SC Malaysia) use related principles but different financial thresholds, denominators, averaging windows, liquidity tests, and implementation rules \citep{AAOIFI2024,DJIM2024,SP2024,MSCI2024,FTSE2022,SCMalaysia2013}. The same company can therefore be investable for one pool of Shariah-constrained capital and excluded for another even when its underlying accounts are unchanged.

Current permission is only the first finance object. A binary pass label does not reveal how stable that permission is. Two firms can both be eligible today even though one lies comfortably inside every relevant debt, cash, receivables, or liquidity threshold while the other lies marginally inside its binding screen. The first can absorb modest changes in its accounts or market-value denominator without leaving the feasible set; the second is more exposed to a future reclassification. Current eligibility and classification-boundary risk are therefore economically different states.

A reclassification changes more than the name attached to a security. It can alter the set of mandate-constrained investors formally permitted to hold it. The proposed finance sequence is
\[
\begin{gathered}
\text{classification change}
\;\longrightarrow\;
\text{change in the permitted investor set} \\
\;\longrightarrow\;
\text{portfolio rebalancing}
\;\longrightarrow\;
\text{potential price pressure}.
\end{gathered}
\]
Each arrow is conditional. The affected rule must govern economically relevant capital, the event must not be fully anticipated, and permitted investors must be marginal and able to trade. A formal label change can therefore leave prices unchanged.

This distinction creates a potential state variable for portfolio management before an official review occurs. Current classifications describe what a constrained manager may hold today; cross-standard disagreement and boundary distance indicate where that feasible universe is fragmented or vulnerable to change. Such information can support monitoring of possible forced-divestment exposure, future entrants, portfolio turnover, and liquidity needs around review dates. The paper does not claim that these measures already constitute a profitable trading strategy. Its central research question is narrower: \emph{When does a change in Shariah investability alter the relevant investor base sufficiently to have price content?}

We call disagreement over, and instability in, binary Shariah investability \emph{classification uncertainty}. This object differs from disagreement about expected cash flows, analyst forecasts, or continuous sustainability scores. A Shariah screen determines whether a mandate-governed investor is formally permitted to hold the security at all. Classification uncertainty consequently has two dimensions: cross-sectional screening-rule disagreement over the current eligible set, and dynamic classification-boundary risk over the stability of that set.

The theoretical framework connects these dimensions to investor-base segmentation. A firm is represented by an eligibility vector $e_i$, whose entries indicate whether it is permitted under each Shariah standard. Investors are partitioned across mutually exclusive mandate types. The key economic state variable is the firm's permitted investor mass,
\begin{equation}
    m_i = \omega_U + \omega^\top e_i,
\end{equation}
where $\omega_U$ denotes unconstrained capital and $\omega$ contains the capital masses governed by the different Shariah policy types. In a transparent mean--variance benchmark, required expected returns decline as $m_i$ expands. The framework does not equate permission with actual ownership: recognition, benchmark weights, available cash, beliefs, and trading frictions determine whether formally permitted capital becomes marginal demand. It instead disciplines the relevant comparison. Counting rulebooks is not equivalent to measuring the capital they govern, and two firms with the same number of pass labels can face different investor-base consequences.

The model also turns static Shariah compliance into a dynamic financial-risk state. For every standard, we measure the distance between the firm's active financial ratios and the nearest binding threshold. Boundary proximity captures the stability of current investability before a label changes. A boundary crossing then changes the eligibility vector and, in the model, permitted investor mass. This sequence---boundary proximity, transition risk, formal reclassification, and possible investor-base adjustment---connects portfolio monitoring to the event-study tests. It does not imply that every crossing is priced or that the empirical proximity transformation is a structural transition probability.

The empirical analysis uses two complementary institutional settings. Route A constructs researcher-emulated Shariah classifications for 13,188 U.S. common-stock securities in a 1999--2024 CRSP--Compustat panel containing 1,342,606 security-months. The seven implementations cover AAOIFI, DJIM, S\&P, FTSE/Yasaar, MSCI main, MSCI M-Series, and SC Malaysia. Screening-rule disagreement and proximity to active boundaries rank next-month screen-implied transitions: the corresponding logit coefficients are 1.1609 and 3.2384 under two-way security-and-month clustering. These are monitoring results, not causal estimates of official-provider decisions.

Route A also establishes an important boundary condition. Equal-weighted eligibility does not behave as the structural permitted-investor-mass variable. Its positive baseline Fama--MacBeth association is absorbed by profitability, investment, and screening-ratio controls. A September 2023 DJIM/S\&P methodology change mechanically admits 313 U.S. firms, including 233 under both rulebooks, but produces no robust matched repricing. Placebo failures in 2020 and 2022 further limit the event design. Formal Shariah eligibility can therefore expand without generating observable price effects when the affected capital is not measured or is not marginal.

Route B studies 25 official semi-annual SC Malaysia lists between November 2013 and November 2025. This setting differs from the U.S. researcher-emulated panel because the public authority issues a common official classification used in the local Islamic capital market. Events are defined at the official stock-code level, and day 0 is the release date printed in each PDF. Removing 172 securities first observed only after the preceding review leaves 410 continuously listed inclusions with positive but imprecise matched returns. Applying a turnover floor constructed exclusively from pre-event trading data leaves 295 inclusions with 1.76 percentage points over $[0,10]$ and 2.25 points over $[0,20]$, significant under list-date clustered and wild-cluster inference. The estimates remain positive in every leave-one-date-out replication and in the first-inclusion-only sample.

The Malaysia evidence remains conditional. Larger complete-case estimates reflect sample composition, not the addition of fundamental controls. Cumulative pre-event windows are insignificant and a three-month mid-review placebo is null, but a joint test of 20 daily pre-event coefficients rejects equality to zero. The inclusion slope with respect to turnover is approximately zero, strict same-industry matching weakens precision, and dated LSEG ownership changes do not identify a unique constrained marginal buyer. The evidence is therefore a conditional event-return pattern, not a clean causal demand shock.

The paper makes three connected contributions. First, it defines Shariah classification uncertainty as a finance object distinct from a current binary pass label. Cross-standard disagreement measures fragmentation in current investability, while boundary proximity measures its dynamic stability. Second, it formalizes permitted investor mass and shows why the economic importance of a Shariah label depends on the constrained capital governed by the relevant rule rather than the number of approving standards. Third, it links ex ante classification-boundary risk to realized official reclassification events and shows empirically that formal permission is not sufficient for repricing: the institutional setting and the ability of investors to trade matter, although institutional recognition itself is not separately identified.

Taken together, the framework has a direct portfolio interpretation. Binary classifications map the securities a constrained investor may hold today; classification-boundary risk maps where that feasible universe is most exposed to change. For asset managers, the first object determines current compliance, whereas the second identifies prospective rebalancing, forced-divestment exposure, and event risk. By linking boundary proximity to screen-implied transitions and official classification events to subsequent returns, the paper provides a foundation for treating investability risk as a portfolio-monitoring input rather than a static screening output. Establishing whether an ex ante strategy can monetize that information requires a separately validated forecasting and trading design.

Shariah screening is valuable for general finance because it makes investor-set restrictions unusually transparent. Its classifications are binary, its financial boundaries are observable, and its institutional implementations vary across standards and countries. In economic structure, a Shariah reclassification resembles other investability events---such as benchmark eligibility or regulatory admission---in which an institutional rule changes who can hold or track a security. The mechanism can therefore inform research on green taxonomies, controversial-activity exclusions, and sustainability benchmarks. The broader implication follows from the Shariah evidence rather than replacing it: institutional classifications can affect asset prices when they alter the set of investors able and willing to hold the security.

The remainder of the paper proceeds as follows. Section \ref{sec:lit} reviews the institutional setting and related literature. Section \ref{sec:theory} develops the theoretical framework. Section \ref{sec:data} describes the data and measurement. Section \ref{sec:design} presents the empirical designs. Section \ref{sec:us_results} reports Route A results. Section \ref{sec:malaysia} reports Route B evidence. Section \ref{sec:interpretation} integrates the findings and discusses implications. Section \ref{sec:conclusion} concludes.


\section{Institutional Background and Related Literature}\label{sec:lit}

\subsection{Shariah equity screening and institutional heterogeneity}

Shariah-compliant equity screening converts religious-legal investment principles into operational security-selection rules. Recognized methodologies typically combine exclusions for prohibited business activities with financial screens that limit interest-bearing debt, interest-bearing assets, receivables, liquidity, and impermissible income \citep{KhatkhatayNisar2007,AyedhEchchabi2019,RizaldyAhmed2019}. The resulting classification is binary for portfolio implementation: under a given standard, a security is either inside or outside the eligible investment set.

The broad objective is shared, but implementation is heterogeneous. AAOIFI, DJIM, S\&P, FTSE/Yasaar, MSCI, and SC Malaysia differ in denominator choice, market-capitalization averaging windows, treatment of cash and receivables, and threshold architecture \citep{AAOIFI2024,DJIM2024,SP2024,MSCI2024,FTSE2022,SCMalaysia2013}. A debt ratio divided by total assets need not produce the same classification as one divided by a 24- or 36-month average market capitalization. A liquidity screen used by one provider may be absent from another. These are observable institutional differences that can make the same security investable under one recognized methodology and ineligible under another.

Prior Islamic-finance research studies screening criteria, Islamic-index performance, portfolio construction, and continuous compliance rankings \citep{HoRahmanYusufZamzamin2014,AshrafMohammad2016,AlnamlahHassan2022,OrhanIsiker2021}. The present question is different. It does not rank firms by religious merit or interpret a larger financial margin as being ``more Shariah compliant.'' Religious-legal eligibility remains standard-specific and binary. The financial object is the stability and investor-base robustness of that eligibility: how many recognized rules admit the security, how close it is to losing or gaining eligibility, and how much constrained capital is governed by the affected classification.

\subsection{Investor-base segmentation and constrained demand}

The economic mechanism follows the finance literature on limited participation, investor recognition, and segmented demand. Securities with narrower investor recognition can require higher expected returns because risk sharing is incomplete \citep{Merton1987}. When norms or mandates prevent part of the market from holding an asset, the remaining investors must absorb its supply \citep{HeinkelKraus2001,HongKacperczyk2009,FabozziMa2008}. Sustainable-finance models similarly show that tastes and portfolio constraints can affect prices and holdings \citep{PastorStambaugh2021,PedersenFitzgibbons2021,Zerbib2022}. Related work emphasizes that effects on the cost of capital depend on whether constrained investors are economically large and marginal \citep{BakerHollifieldOsambela2022,BerkvanBinsbergen2025,FeldhutterPedersen2025}.

Shariah screening changes the object of segmentation. The constraint is not assumed to be common knowledge and uniform across investors; it is implemented through heterogeneous rulebooks and institutional authorities. The same security can therefore face several overlapping feasible sets. Formal admission expands a potential investor set, but it need not move prices if the affected investors do not recognize the change, are not marginal, cannot trade the security efficiently, or anticipated the event. Evidence that labels can coordinate fund flows supports this distinction between formal classification and implemented demand \citep{HartzmarkSussman2019,CeccarelliRamelliWagner2024}. The event studies accordingly test a conditional implication: a change in Shariah investability can have price content when it alters an economically relevant investor base, not simply because a label changes.

\subsection{Classification disagreement and the paper's contribution}

Research on ESG ratings establishes that providers disagree because they differ in scope, measurement, and weighting \citep{BergKolbel2022,ChristensenSerafeim2022,GibsonBrandon2021}. That disagreement can be associated with returns, information, financing outcomes, and equilibrium demand \citep{AvramovCheng2022,WangWangDongWang2024IRFA,HuLiLi2024IRFA,HePanShanZhou2025IRFA}. This literature provides a useful measurement analogy, but the institutional object here is sharper. ESG disagreement usually concerns continuous scores or rankings within an investable universe. Shariah screening disagreement changes the eligible universe itself. Its thresholds and denominators are also comparatively transparent, allowing disagreement to be connected to specific rule components and subsequent binary transitions.

The paper fills three related gaps. It first defines cross-standard Shariah classification uncertainty as fragmentation in current investability. It then links eligibility vectors to permitted investor mass, distinguishing equal rulebook counts from capital-weighted economic exposure. Finally, it adds a dynamic dimension by measuring distance to active classification boundaries and connecting ex ante instability to realized rulebook and official-list events. This combination moves the Shariah-screening literature from static compliance and index performance toward a mainstream finance question: can the stability of institutional permission serve as a state variable for constrained participation, portfolio monitoring, and asset prices?

\section{Theory: Classification Uncertainty and Permitted Investor Mass}\label{sec:theory}

\subsection{Economic setting and notation}

The model treats recognized Shariah screening policies as rules that partition constrained investors into different feasible asset sets. Assets are indexed by $i$, standards by $s$, and investor types by the implemented policy that governs their holdings. The purpose is to isolate how formal eligibility changes potential market participation; it is not a model of religious adjudication or a claim that all permitted investors hold identical portfolios. The main notation used throughout the theoretical and empirical analysis is summarized in Table~\ref{tab:notation}. In the theory, $N$ denotes the number of risky assets. In empirical tables, $N$ retains its conventional meaning as the number of observations or matched events.

\begin{table}[!htbp]
\centering
\caption{Main notation}
\label{tab:notation}
\begin{threeparttable}
\footnotesize
\begin{tabularx}{\textwidth}{L{0.18\textwidth}X}
\toprule
Symbol & Meaning \\
\midrule
$i$, $s$, $k$, $t$, $d$ & Firm/security, screening standard, financial screen, monthly/list-date index, and daily trading-day index. \\
$e_i^{(s)}$ & Binary eligibility indicator: one if firm $i$ is permitted under standard $s$, zero otherwise. \\
$e_i$ & Eligibility vector collecting all $S$ standard-specific indicators for firm $i$. \\
$\bar e_i$ & Unweighted eligibility share, equal to the fraction of standards that permit firm $i$. \\
$\Disc_i$ and $\Disc_i^\omega$ & Unweighted and capital-weighted disagreement measures. Both equal zero under full consensus and are largest when eligible and ineligible rulebook mass is evenly split. \\
$\omega_s$, $\omega_U$, $\omega$ & Mutually exclusive capital mass governed by standard $s$, unconstrained capital mass, and the vector of standard-specific capital masses. \\
$m_i$ & Potential permitted investor mass: the total capital mass formally allowed to hold firm $i$, before recognition, tastes, benchmarks, and trading frictions determine actual demand. \\
$\tilde r_i$, $\mu_i$, $\tilde\epsilon_i$ & Random one-period excess return, its expected excess return, and its zero-mean return shock. Empirical sections use realized returns $r_{i,t}$ or $r_{i,d}$. \\
$\gamma$, $\sigma^2$, $\Sigma$ & Common risk aversion, common idiosyncratic return variance under the benchmark, and the general return covariance matrix. \\
$q_i$, $D_i$ & Per-unit-capital desired holding of asset $i$ and aggregate demand for asset $i$. \\
$\Delta_{i,s,t}$, $D_{i,t}$, $P_{i,t}$ & Standard-specific active margin, nearest active-boundary distance, and its monotone proximity-risk transformation. \\
$AR$, $\CAR$, $\diffCAR$ & Abnormal return, cumulative abnormal return, and matched treatment-minus-control cumulative abnormal return. \\
$\ell_t$, $a_t$, $\mathcal{L}_{i,t}$, $Q^{turn}_{0.25}$ & SC list date, official PDF release date, liquidity-screen indicator, and bottom-quartile cutoff of pre-event turnover in the Route B event universe. \\
\bottomrule
\end{tabularx}
\begin{tablenotes}
\footnotesize
\item The notation deliberately separates theory objects from empirical proxies. The structural theory prices $m_i$ with mandate-capital weights $\omega_s$; the U.S. diagnostics test implications consistent with the framework using equal rulebook weights because historical mandate-weighted capital is not observed at the required security-month frequency.
\end{tablenotes}
\end{threeparttable}
\end{table}

\subsection{Eligibility disagreement and permitted investor mass}

There are $N$ risky assets indexed by $i=1,\ldots,N$ and $S$ screening standards indexed by $s=1,\ldots,S$. Standard $s$ maps firm $i$ into a binary eligibility indicator
\begin{equation}
    e_i^{(s)}\in\{0,1\}.
\end{equation}
Collect these indicators into an eligibility vector
\begin{equation}
    e_i = \left(e_i^{(1)},e_i^{(2)},\ldots,e_i^{(S)}\right)^\top .
\end{equation}
The unweighted eligibility share is
\begin{equation}
    \bar e_i = \frac{1}{S}\sum_{s=1}^S e_i^{(s)}.
\end{equation}
A convenient unweighted disagreement measure is
\begin{equation}
    \Disc_i = 4\bar e_i(1-\bar e_i),
\end{equation}
which equals zero under full consensus and is maximized when standards are evenly split.

When capital weights are known, the corresponding capital-weighted disagreement measure replaces the simple rulebook share with the eligible share of constrained rulebook capital:
\begin{equation}
    \bar e_i^\omega =
    \frac{\sum_{s=1}^{S}\omega_s e_i^{(s)}}{\sum_{s=1}^{S}\omega_s},
    \qquad
    \Disc_i^\omega = 4\bar e_i^\omega(1-\bar e_i^\omega).
\end{equation}
The unweighted measure $\Disc_i$ is therefore a special case of $\Disc_i^\omega$ when standards receive equal weights. In the empirical sections, $\Disc_{i,t}$ denotes the implemented Route A disagreement variable; because the Route A diagnostics use equal rulebook weights, this implemented disagreement coincides with the raw split-disagreement measure in the preceding equation.

Let $\omega_s>0$ be the mass of capital governed by standard $s$, and let $\omega_U>0$ be unconstrained capital. Investor types are mutually exclusive, so a mandate following several source documents must be assigned to one implemented policy type rather than counted once under each source. This prevents double counting. Normalize
\begin{equation}
    \omega_U + \sum_{s=1}^S \omega_s = 1.
\end{equation}
The key economic object is potential permitted investor mass:
\begin{equation}
    m_i = \omega_U + \omega^\top e_i = \omega_U + \sum_{s=1}^{S}\omega_s e_i^{(s)}.
    \label{eq:mass}
\end{equation}
This is the potential mass of investors formally allowed to hold firm $i$. It is not actual ownership or an assertion that every permitted investor trades. Recognition, tastes, benchmark weights, available cash, and trading frictions determine how much of this potential mass becomes demand. A rulebook count is not enough; a label matters through the distinct capital whose implemented policy it governs.

\subsection{Pricing under segmented investor bases}

Asset $i$ has random one-period excess return
\begin{equation}
    \tilde r_i = \mu_i + \tilde\epsilon_i,
\end{equation}
where $\mu_i=\E[\tilde r_i]$ is the required expected excess return and $\tilde\epsilon_i$ is a zero-mean return shock. The return-shock vector $\tilde\epsilon$ has covariance matrix $\Sigma$. Investors have common mean--variance preferences with risk aversion $\gamma>0$. Unconstrained investors can hold all assets. A constrained investor following standard $s$ can hold asset $i$ only if $e_i^{(s)}=1$.
	
	\begin{assumption}[Investor-base segmentation]
	Investor type $U$ controls normalized wealth mass $\omega_U$ and can hold all assets. Each constrained mandate belongs to one mutually exclusive type $s$, controls wealth mass $\omega_s$, and can hold asset $i$ only when $e_i^{(s)}=1$. Type masses are exogenous over the event window, investors share beliefs and risk aversion, and the only difference across types is the admissible set.
	\end{assumption}
	
	\begin{assumption}[Symmetric covariance benchmark]
	Return shocks satisfy $\Sigma=\sigma^2 I$, where $\Sigma$ is the $N\times N$ return-shock covariance matrix, $I$ is the identity matrix, every asset has the same variance $\sigma^2$, and every off-diagonal covariance is zero. Each risky asset has normalized net supply $1/N$.
	\end{assumption}

Assumption 2 is a tractability benchmark, not an empirical claim that stock returns are uncorrelated or equally volatile. Equal variances and zero covariances make the demand for each admissible asset separable, allowing the investor-base mechanism to be expressed in a transparent closed form. With a general covariance matrix, Shariah eligibility still changes each investor type's feasible set, but demand for asset $i$ also depends on its covariance with every other admissible asset. Permitted investor mass would then remain economically relevant without being sufficient for the same one-line equilibrium expression.

Holdings in the benchmark are measured as portfolio weights per unit of investor-type wealth, and total asset supply is normalized to $1/N$. The normalization is without loss for the sign predictions: changing the supply normalization rescales equilibrium expected excess returns but does not change the comparative static with respect to permitted investor mass.

\begin{proposition}[Permitted investor mass and required returns]\label{prop:mass}
Under the symmetric covariance benchmark, the equilibrium expected excess return of asset $i$ is
\begin{equation}
    \mu_i^* = \frac{\gamma\sigma^2}{N}\frac{1}{m_i}.
\end{equation}
Thus required expected returns are decreasing in potential permitted investor mass.
\end{proposition}

\begin{proof}
For one unit of capital of investor type $h$, let $A_h$ denote its admissible asset set and $q_h$ its vector of risky-asset holdings. A risk-free asset absorbs residual wealth. The type solves
\[
    \max_{q_h:\,q_{h,i}=0\ \forall i\notin A_h}
    q_h^\top\mu-\frac{\gamma}{2}q_h^\top\Sigma q_h .
\]
Under $\Sigma=\sigma^2I$, the first-order condition for every admissible asset is
$q_{h,i}=\mu_i/(\gamma\sigma^2)$; demand is zero for inadmissible assets. Because covariance is diagonal, this vector problem separates across admissible assets without imposing an additional risky-asset budget constraint. Unconstrained capital can always hold asset $i$, while constrained capital of type $s$ can hold it only if $e_i^{(s)}=1$. Aggregate demand is therefore
\begin{equation}
	    D_i = \frac{\mu_i}{\gamma\sigma^2}\left(\omega_U+\sum_{s=1}^{S}\omega_s e_i^{(s)}\right)
	        = \frac{\mu_i}{\gamma\sigma^2}m_i.
\end{equation}
Market clearing requires $D_i=1/N$. Solving gives the result.
\end{proof}

\begin{remark}[Asset-specific supply and variance]\label{rem:supply_variance}
The inverse relationship between required expected return and permitted investor mass does not rely on equal supply or common variance. If asset $i$ has net supply $S_i>0$ and idiosyncratic return variance $\sigma_i^2$, the same one-asset clearing argument gives
\begin{equation}
    \mu_i^*=\frac{\gamma\sigma_i^2 S_i}{m_i}.
\end{equation}
Higher supply, larger variance, or lower permitted investor mass raises the required expected excess return. Empirical size and volatility measures proxy the model's supply and risk dimensions. Liquidity variables and matching diagnostics instead address whether demand can be implemented and whether treated--control comparisons are credible; they do not enter this comparative static directly.
\end{remark}

\begin{corollary}[Average eligibility is insufficient]\label{cor:elig}
Two firms with the same unweighted eligibility share can have different expected returns if their eligibility vectors load differently on $\omega$.
\end{corollary}

\begin{corollary}[Unweighted disagreement is insufficient]\label{cor:disagree}
Two firms can have the same unweighted disagreement but different expected returns and investor-base exposure if $\omega^\top e_i\neq \omega^\top e_j$.
\end{corollary}

The intuition is direct. A firm permitted by standards that govern a large amount of constrained capital has a broader effective investor base than a firm admitted by the same number of small standards. Average eligibility and unweighted disagreement count labels; $m_i$ weights labels by economic relevance.

\begin{remark}[Potential permission is not ownership]\label{rem:ownership}
Equation~\eqref{eq:mass} describes admissibility, not equilibrium ownership shares. Actual ownership additionally depends on heterogeneous beliefs, benchmark weights, portfolio constraints, and whether permitted investors recognize and trade the security. The model therefore signs a price comparative static only under the benchmark assumptions; it does not yield a separate concentration or liquidity proposition. The dated holder tests later in the paper are empirical mechanism diagnostics rather than identities implied by $m_i$.
\end{remark}

\subsection{Classification-boundary risk and Shariah investability}\label{sec:boundary_risk}

Binary Shariah eligibility describes whether a security is investable under a standard today; it does not describe how stable that investability is. Consider two firms that currently pass the same rulebook. One lies comfortably inside every relevant debt, cash, receivables, and income threshold. The other lies only marginally inside its binding debt threshold. Their current eligibility indicators are identical, but their exposure to a future classification change is not.

This distinction matters for constrained portfolio choice. A security close to an active boundary creates a greater need for monitoring and can expose a mandate to future rebalancing, transaction costs, benchmark turnover, and a change in the set of investors permitted to hold it. For a portfolio manager, the current label determines today's feasible set, whereas boundary distance identifies where that set is vulnerable to future change. Boundary risk is not a statement that the firm is more or less Shariah compliant in a religious or legal sense, nor is proximity by itself a trading signal. It is a financial measure of the stability of the binary classification produced by an implemented standard.

To make that stability observable, let $K_s$ denote the set of active financial screens used by standard $s$. The normalized active margin for firm $i$ under standard $s$ is
\begin{equation}
    \Delta_{i,s,t} = \min_{k\in K_s}\frac{\theta_{s,k}-x_{i,k,t}}{\theta_{s,k}},
\end{equation}
where $x_{i,k,t}$ is the relevant accounting ratio and $\theta_{s,k}$ is the corresponding threshold. Dividing by $\theta_{s,k}$ places heterogeneous screens on a comparable relative scale. A positive value means that the firm remains inside every active financial boundary under standard $s$; a negative value means that at least one active screen is breached. Because the minimum selects the binding screen, $\Delta_{i,s,t}$ records the standard-specific margin most relevant for a near-term transition.

Across standards, define the nearest active-boundary distance as
\begin{equation}
    D_{i,t} = \min_s |\Delta_{i,s,t}|
\end{equation}
and its monotone proximity-risk transformation
\begin{equation}
    P_{i,t} = \exp(-5D_{i,t}).
\end{equation}
Small $D_{i,t}$ means that at least one standard places the firm close to an active financial boundary. The fixed scale parameter 5 makes the empirical measure decline smoothly with distance; it is neither estimated nor calibrated as a transition-probability parameter. Because the exponential transformation is strictly decreasing, $P_{i,t}$ preserves the ordering implied by $D_{i,t}$. It should therefore be read as a convenient monitoring index, not as $\Prb(\mathrm{Transition}_{i,t+1}=1)$.

\begin{proposition}[A single active boundary and classification transitions]\label{prop:boundary}
	Suppose one lower-is-better screening ratio evolves as $x_{i,k,t+1}=x_{i,k,t}+u_{i,k,t+1}$, where $u$ has a continuous distribution. For a fixed threshold $\theta_{s,k}$, the probability of crossing that boundary is weakly decreasing in current absolute distance from it. A crossing of standard $s$ changes potential permitted mass by $\omega_s$.
\end{proposition}
	
\begin{proof}
For a fixed positive threshold $\theta$, let $d=|\theta-x_{i,k,t}|$ and let $F$ denote the cumulative distribution function of $u_{i,k,t+1}$. If the firm is currently inside the boundary, crossing into breach requires $u_{i,k,t+1}>d$; the crossing probability is $1-F(d)$, which is weakly decreasing in $d$. If the firm is outside the boundary, crossing into compliance requires $u_{i,k,t+1}<-d$; continuity implies that the probability is $F(-d)$, which is also weakly decreasing in $d$ because every cumulative distribution function is non-decreasing in its argument. The eligibility indicator for standard $s$ changes at the crossing, so equation~\eqref{eq:mass} changes by $\omega_s$ in absolute value.
\end{proof}

The proposition formalizes only the local intuition: for one active lower-is-better screen, a firm closer to the threshold is more exposed to a crossing under the stated innovation process. It does not establish that minimum distance fully characterizes transition probabilities when several correlated screens operate simultaneously. In that setting, the probability of any label change depends on the full vector of signed margins and their joint innovation distribution. The empirical $D_{i,t}$ is deliberately parsimonious because it identifies the nearest point of classification fragility without pretending to estimate that multivariate process.

The economic sequence is therefore conditional but clear. Greater boundary proximity indicates a less stable current classification. A realized accounting or methodology change may then cross a boundary and alter formal eligibility. That label change changes potential permitted investor mass by the capital governed by the affected standard. Portfolio turnover and price pressure can follow only if that capital recognizes the classification, is marginal for the security, and can trade. Section~\ref{sec:design} tests the first link as a monitoring ranking, while the event studies examine the realized-label stage. Together, these links convert static Shariah investability into a dynamic financial-risk state without claiming that boundary proximity itself causes accounting changes or asset-price movements.

\subsection{Realized reclassification and permitted investor mass}\label{sec:realized_crossings}

Section~\ref{sec:boundary_risk} concerns instability before a classification changes. The realized counterpart is what happens to the feasible investor set once a financial boundary, methodology rule, or official classification changes. Such an event can alter formal eligibility even without contemporaneous cash-flow news. Let $s^*$ denote the affected standard or authority. The general change in permitted investor mass is
\begin{equation}
    m_i^{post}-m_i^{pre}
    =\omega_{s^*}\left(e_i^{(s^*),post}-e_i^{(s^*),pre}\right).
\end{equation}
For an inclusion, $e_i^{(s^*),pre}=0$ and $e_i^{(s^*),post}=1$, so
\begin{equation}
    m_i^{post}=m_i^{pre}+\omega_{s^*}.
\end{equation}
The capital governed by $s^*$ becomes formally permitted to hold the security. For an exclusion, $e_i^{(s^*),pre}=1$ and $e_i^{(s^*),post}=0$, so
\begin{equation}
    m_i^{post}=m_i^{pre}-\omega_{s^*}.
\end{equation}
The potential investor base contracts by the same mandate-capital mass. These expressions formalize the first link from reclassification to a changed feasible investor set. They describe changes in admissibility, not mechanical purchases or sales: portfolio rebalancing and price pressure arise only through the additional conditions stated below.

Using Proposition~\ref{prop:mass}, the change in the required expected excess return is
\begin{align}
    \Delta \mu_i^*
    &= \frac{\gamma\sigma^2}{N}
    \left(\frac{1}{m_i^{pre}+\omega_{s^*}}-\frac{1}{m_i^{pre}}\right) \nonumber \\
    &= -\frac{\gamma\sigma^2}{N}
    \frac{\omega_{s^*}}{m_i^{pre}(m_i^{pre}+\omega_{s^*})}<0.
\end{align}
If the classification event does not change expected cash flows, and if the expected one-period payoff is $X_i$, then a one-period pricing relation is $p_i=X_i/(1+r_f+\mu_i^*)$, where $r_f$ is the risk-free rate and $\mu_i^*$ is the required excess return. Hence a decline in the required expected excess return implies a positive price reaction. This is the theoretical sign restriction behind the inclusion-event tests; the empirical question is whether the formal eligibility change actually moves marginal constrained demand.

The theoretical change in feasible investor mass is symmetric across inclusion and exclusion, but empirical price responses need not be. The amount of capital governed by the authority can differ across settings; purchases and sales can face different implementation lags; exclusions may be anticipated; short-sale and benchmark constraints can create asymmetry; and other investors may absorb order flow. The model therefore supplies a comparative-static sign under its benchmark assumptions, not a prediction that inclusion and exclusion CARs must have equal and opposite magnitudes.

Five conditions determine whether a realized classification change becomes economically binding. The affected standard must govern a meaningful amount of capital; those investors must be marginal for the security; the event must contain information not already incorporated into price; the authority must be recognized by the relevant mandates; and the security must be sufficiently tradable for portfolios to adjust. These conditions explain the empirical comparison without claiming that the comparison identifies any one condition separately. The September 2023 DJIM/S\&P methodology shock expands formal U.S. rulebook eligibility but does not produce robust matched repricing. Official SC Malaysia inclusions exhibit a stronger conditional return pattern in a local institutional setting, yet the paper does not separately identify recognition or the marginal buyer. Sections 3.4 and 3.5 therefore provide the bridge from ex ante classification risk to realized event returns while preserving the distinction between permission, demand, and price.

\section{Data and Measurement}\label{sec:data}

\subsection{Route A: U.S. cross-standard panel}

Route A combines CRSP monthly stock data, Compustat annual fundamentals, and the CRSP--Compustat Merged link table. We retain ordinary U.S. common shares with CRSP share codes 10 or 11 and NYSE, AMEX, or NASDAQ exchange codes. The analysis sample runs from January 1999 to December 2024. CRSP data are pulled from 1995 onward before constructing 24- and 36-month rolling market-capitalization denominators, so that early analysis-period denominators are not mechanically truncated.

Annual accounting data are matched with a six-month reporting lag, following standard asset-pricing practice \citep{FamaFrench1992,HouXue2015}. Each annual accounting record is used from its availability month until the next available annual record, with a stale-data guard to prevent old fundamentals from being carried forward indefinitely. Table \ref{tab:sample_us} summarizes the sample construction.

\begin{table}[!htbp]
\centering
\caption{Route A: U.S. sample construction}
\label{tab:sample_us}
\begin{threeparttable}
\begin{tabularx}{\textwidth}{Xrr}
\toprule
Step & Observations & Unique securities \\
\midrule
CRSP common shares after 1999 filter & 1,376,235 & 13,718 \\
After CCM linking, valid lagged Compustat merge, link de-duplication, and stale-data guard & 1,342,606 & 13,188 \\
\bottomrule
\end{tabularx}
\begin{tablenotes}
\footnotesize
\item The unit of observation is a security-month, and a security is identified by CRSP permno. The combined second row avoids presenting the lagged-accounting and link-quality conditions as a later filter that removes zero observations. The panel contains 12,998 permcos and 13,202 linked gvkeys; no permno-month is duplicated. The final panel spans January 1999 through December 2024.
\end{tablenotes}
\end{threeparttable}
\end{table}

For each security-month, we construct binary eligibility indicators under seven researcher-emulated rulebooks. These labels are not historical classifications supplied by the standard setters. They apply the documented thresholds in Table~\ref{tab:rulebook_implementation} to a common, lagged Compustat data set and a common SIC-based business-activity exclusion map. This common-data design isolates differences in codified financial screens, but it cannot reproduce proprietary revenue purification, issuer review, or judgment by each provider.

\begin{table}[!htbp]
\centering
\caption{Researcher implementation of the seven financial-screen rulebooks}
\label{tab:rulebook_implementation}
\begin{threeparttable}
\scriptsize
\begin{tabularx}{\textwidth}{L{0.13\textwidth}L{0.18\textwidth}L{0.18\textwidth}L{0.25\textwidth}L{0.12\textwidth}}
\toprule
Implementation & Debt screen & Cash screen & Receivables/liquidity screen & Income screen \\
\midrule
AAOIFI & Debt/market cap $\leq30\%$ & Cash and interest-bearing assets/market cap $\leq30\%$ & None & Proxy/sales $\leq5\%$ \\
DJIM & Debt/24-month average market cap $\leq33\%$ & Cash and interest-bearing assets/24-month average market cap $\leq33\%$ before Aug.\ 2023 & Receivables/24-month average market cap $\leq33\%$ before Aug.\ 2023 & Proxy/sales $\leq5\%$ \\
S\&P & Debt/36-month average market cap $\leq33\%$ & Cash and interest-bearing assets/36-month average market cap $\leq33\%$ before Aug.\ 2023 & Receivables/36-month average market cap $\leq33\%$ before Aug.\ 2023 & Proxy/sales $\leq5\%$ \\
FTSE/Yasaar & Debt/assets $\leq1/3$ & Cash/assets $\leq1/3$ & (Cash+receivables)/assets $\leq50\%$ & Proxy/sales $\leq5\%$ \\
MSCI main & Debt/assets $\leq1/3$ & Cash/assets $\leq1/3$ & (Cash+receivables)/assets $\leq70\%$ & Proxy/sales $\leq5\%$ \\
MSCI M-Series & Debt/36-month average market cap $\leq1/3$ & Cash/36-month average market cap $\leq1/3$ & (Cash+receivables)/36-month average market cap $\leq49\%$ & Proxy/sales $\leq5\%$ \\
SC Malaysia & Debt/assets $\leq33\%$ & Cash/assets $\leq33\%$ & None in the implemented financial-ratio layer & Proxy/sales $\leq5\%$ \\
\bottomrule
\end{tabularx}
\begin{tablenotes}
\footnotesize
\item Debt is total debt; cash includes cash and positive short-term investments. The baseline non-permissible-income proxy is $[\max(\mathrm{IDIT},0)+\max(\mathrm{NOPI},0)]/\mathrm{SALE}$. Because Compustat's NOPI also contains permissible non-operating items, this is a conservative, overinclusive proxy rather than an official purification measure. Missing or nonpositive required denominators fail the relevant screen. The Internet Appendix rebuilds all seven labels, disagreement, proximity, next-month transitions, and Fama--MacBeth characteristics using a narrower positive-IDIT-only proxy.
\end{tablenotes}
\end{threeparttable}
\end{table}

Table \ref{tab:elig_rates} and Figure \ref{fig:elig_bar} report average eligibility rates under the baseline implementation.

\begin{table}[!htbp]
\centering
\caption{Average eligibility rates by standard}
\label{tab:elig_rates}
\begin{tabular}{lr}
\toprule
Standard & Mean eligibility \\
\midrule
FTSE/Yasaar & 25.60\% \\
S\&P & 26.66\% \\
MSCI M-Series & 26.69\% \\
DJIM & 27.30\% \\
MSCI main & 28.37\% \\
SC Malaysia & 29.26\% \\
AAOIFI & 29.29\% \\
\bottomrule
\end{tabular}
\end{table}

\begin{figure}[!htbp]
\centering
\includegraphics[width=0.84\textwidth]{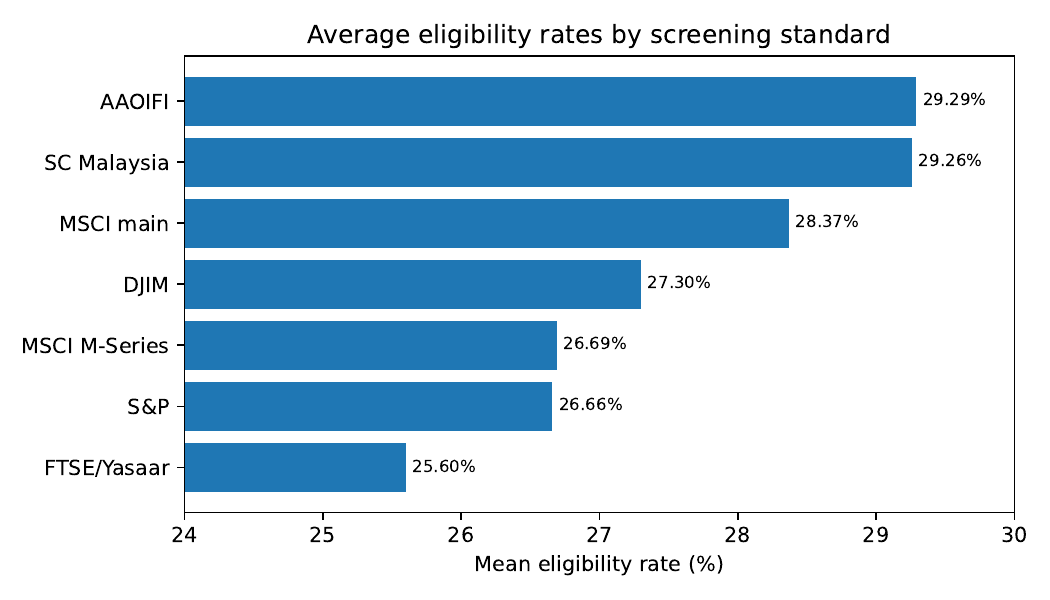}
\caption{Average eligibility rates by implemented screening standard. The spread across standards illustrates that permitted universes differ materially even when the broad Shariah objective is common.}
\AltText{Horizontal bar chart ranking seven Shariah screening standards by average eligibility rate. The rates range from about 25.6 percent for FTSE/Yasaar to about 29.3 percent for AAOIFI, showing modest but economically meaningful differences across rulebooks.}
\label{fig:elig_bar}
\end{figure}
\FloatBarrier

The main Route A variables are eligibility share $\bar e_{i,t}$, implemented disagreement $\Disc_{i,t}$, active-boundary distance $D_{i,t}$, and proximity risk $P_{i,t}=\exp(-5D_{i,t})$. With seven equal rulebook weights, implemented disagreement is $4\bar e_{i,t}(1-\bar e_{i,t})$. We do not label any equal-weighted count as investor mass: without historical mandate weights, $\bar e_{i,t}$ is a rulebook-coverage measure rather than an estimate of the structural object $m_i$.

To connect the structural object $m_i$ to the equal-weighted empirical proxy, Table~\ref{tab:weight_sensitivity} reports a mandate-weight sensitivity exercise. We recompute $\sum_s\omega_s e_{i,t}^{(s)}$ under transparent scenario weights: equal rulebook weights, DJIM/S\&P-heavy weights, MSCI-heavy weights, SC-heavy weights, and AAOIFI-heavy weights. These tilts are deliberately severe stress tests, not estimates of historical capital shares. They ask whether the equal-weighted ranking is fragile to concentrated mandate capital in one authority or index family. The answer is no for the main ranking: scenario scores have Spearman correlations of 0.992--1.000 with the equal-weighted score, mean within-month rank shifts of at most 1.3 percentage points, and top-quintile overlap above 87\%. The bottom-quintile overlap is mechanically high because all-excluded firms remain all-excluded under every scenario.

\input{table_routea_weight_sensitivity}

\subsection{September 2023 DJIM/S\&P methodology shock}

The Route A event study uses the August 2023 pre-event classification state. On August 4, 2023, S\&P Dow Jones Indices announced changes to the compliance criteria for the Dow Jones Islamic Market and S\&P Shariah index families. The changes removed cash and receivables screens while retaining other screens, and were first visible in pro-forma files beginning September 1, 2023 \citep{SPDJIDJIMAnnouncement2023,SPDJIShariahAnnouncement2023}. Firms that failed the old removed screens but passed the retained screens become mechanically eligible. Table \ref{tab:sep2023_groups} reports the treatment groups.

\begin{table}[!htbp]
\centering
\caption{September 2023 DJIM/S\&P treatment groups}
\label{tab:sep2023_groups}
\begin{tabular}{lrr}
\toprule
Group & Firms & Share \\
\midrule
Still ineligible & 2,607 & 66.64\% \\
Continuously eligible & 992 & 25.36\% \\
Newly eligible & 313 & 8.00\% \\
Newly eligible under DJIM & 289 & 92.33\% of newly eligible \\
Newly eligible under S\&P & 257 & 82.11\% of newly eligible \\
Newly eligible under both DJIM and S\&P & 233 & 74.44\% of newly eligible \\
DJIM-only newly eligible & 56 & 17.89\% of newly eligible \\
S\&P-only newly eligible & 24 & 7.67\% of newly eligible \\
\bottomrule
\end{tabular}
\end{table}

\subsection{Route B: SC Malaysia official-list events}

Route B uses all available SC Malaysia Shariah-compliant securities lists from November 2013 to November 2025 \citep{SCMalaysiaLists2013_2025}. The final parser covers 25 semi-annual list dates. We extract stock codes and company names from official PDF lists, validate full-list counts, and construct a canonical membership panel with one row per list date and SC stock code. The full compliant list is the authoritative event source. Direct inclusion and exclusion tables in the PDFs are used only for validation because their layout changes over time and older documents do not contain symmetric direct-change tables.\footnote{Our parser recovers 633 of an officially reported 653 stock codes from the November 2013 release; the 20-name shortfall reflects pagination and column-break artifacts in the earliest PDF format, which used a different layout than the post-2014 documents. We retain November 2013 as the baseline anchor for the May 2014 transition. Excluding the November 2013 list from the sample and re-anchoring on May 2014 leaves all reported event-study estimates economically and statistically unchanged; the audit comparison of parsed-versus-official counts for every release is reported in the replication archive.}

The Malaysia PDFs contain two empirically distinct dates. The \emph{list date} $\ell_t$ identifies the semi-annual membership period and is the date used to define transitions between membership lists. The \emph{release date} $a_t$ is the publication date printed on the official PDF cover page. For event-study timing, day 0 is the first security-specific trading day on or after $a_t$. The list date $\ell_t$ is retained for matching, membership-state definitions, and clustered inference because all firms in a semi-annual list share the same classification cycle. The replication archive records the release-date override table used for the 25 PDFs.

We define inclusion, exclusion, and continuous-compliance events at the official stock-code level before matching to Compustat Global securities. This order is central: defining events after security matching can turn linkage instability into artificial transitions. A second distinction is equally important. An official-list inclusion may be either a status change for an already traded security or the first appearance of a security listed after the preceding SC review. We identify the latter using the Compustat security's first observed trading date and exclude them from the primary reclassification sample. Table~\ref{tab:routeb_sample_flow} reconciles every reported sample count.

\input{table_routeb_sample_flow}
\FloatBarrier

We use LSEG Workspace in three ways. First, we retain contemporaneous investor-base salience measures as descriptive diagnostics for the broader Malaysia universe. The RIC crosswalk covers 930 of 940 candidate Malaysia securities in the broader cache and 859 securities with event-sample LSEG flags. Second, for the mechanism tests in Section~\ref{sec:malaysia}, we pull dated historical holder snapshots around SC list events. The historical treated-event cache covers 945 RIC-mapped inclusion and exclusion events and contains 55,180 holder rows. We also pull 178,058 historical holder rows for the continuously compliant matched controls used in the release-date-corrected event study. Holder names are classified into Shariah-sensitive, Malaysia state/pension, global passive, and sovereign/Gulf categories using transparent name-based rules, and portfolio-only versions exclude likely strategic or control holders. Third, we pull LSEG ETF NAV, shares outstanding, and current holdings for Malaysia Shariah ETFs around each SC release cycle. These ETF data support aggregate channel-capacity and creation-flow diagnostics, but they do not provide historical security-level benchmark weights. The ETF tests are therefore interpreted as a boundary test for listed ETF demand, not as a complete predicted-purchase design.

\section{Empirical Design}\label{sec:design}

\subsection{Monitoring validation}

The first empirical test asks whether classification uncertainty ranks future screen-implied eligibility changes:
\begin{equation}
    \Prb(\mathrm{Transition}_{i,t+1}=1)=
    \Lambda\left(\alpha+\beta_1\Disc_{i,t}+\beta_2P_{i,t}+\Gamma'X_{i,t}+\delta_{y(t)}\right),
\end{equation}
where $\Lambda(\cdot)$ is the logistic cumulative distribution function, $\mathrm{Transition}_{i,t+1}$ equals one if any researcher-emulated standard-specific label changes in the next consecutive calendar month, $P_{i,t}$ is proximity risk, $X_{i,t}$ includes size, $\Gamma$ is the associated control-coefficient vector, and $\delta_{y(t)}$ are calendar-year fixed effects. Standard errors are two-way clustered by CRSP security and calendar month. The outcome is not an official provider reclassification, and the test is not causal: threshold proximity is mechanically related to a screen-implied transition. The goal is validation of a monitoring ranking.

\subsection{Fama--MacBeth diagnostics}

Proposition \ref{prop:mass} implies that expected returns decline with potential permitted investor mass when $m_i$ captures mutually exclusive mandate-capital weights. Because those weights are unavailable over the full U.S. sample, we estimate reduced-form Fama--MacBeth diagnostics using equal-weighted rulebook characteristics:
\begin{equation}
    r_{i,t+1}-r_{f,t+1}=a_t+b_t Z_{i,t}+\Gamma_t'X_{i,t}+\varepsilon_{i,t+1},
\end{equation}
where $a_t$ is the month-$t$ intercept, $b_t$ is the month-$t$ slope on the diagnostic variable, $Z_{i,t}$ is alternatively eligibility share or implemented disagreement, $X_{i,t}$ is the control vector, $\Gamma_t$ is the month-specific control-coefficient vector, and $\varepsilon_{i,t+1}$ is the next-month residual. The baseline controls are log market capitalization, log book-to-market, 12-to-2 momentum, and one-month reversal. Because Shariah eligibility is partly determined by balance-sheet ratios, expanded specifications add operating profitability, ROA, asset growth, capital expenditure over lagged assets, and the leverage, cash, receivables, and income-proxy ratios used in the screening data. The main variables are standardized before estimation, so coefficients measure the next-month excess-return association for a one-standard-deviation change in the diagnostic variable. We report time-series average coefficients with Newey--West standard errors using 12 lags. These regressions describe equal-weighted formal permission; they are not structural estimates of mandate-weighted demand.

\subsection{Route A event study}

The September 2023 DJIM/S\&P event study uses CRSP daily stock returns and CRSP daily market index returns. Let $d$ index trading days around the event, let $r_{i,d}$ be the daily return of firm $i$, and let $r_{m,d}$ be the CRSP daily market return. For each event date, abnormal returns are estimated using a market model,
\begin{equation}
	    r_{i,d}=\alpha_i+\beta_i r_{m,d}+\epsilon_{i,d},
\end{equation}
estimated over trading days $[-250,-30]$ relative to the event. The fitted $\widehat{\alpha}_i$ and $\widehat{\beta}_i$ are firm-specific market-model parameters. The abnormal return is $\widehat{AR}_{i,d}=r_{i,d}-\widehat{\alpha}_i-\widehat{\beta}_ir_{m,d}$, and the event-window CAR is the conventional sum
\begin{equation}
    \CAR_i(\tau_1,\tau_2)=\sum_{d=\tau_1}^{\tau_2}\widehat{AR}_{i,d}.
\end{equation}
CARs are winsorized at the 1st and 99th percentiles within each event-window cross-section. The two event dates are the August 4, 2023 announcement date and the September 1, 2023 pro-forma visibility date. The preferred treatment group is firms newly eligible under both DJIM and S\&P. The preferred controls are matched still-ineligible firms. We report matched continuously eligible controls and unmatched newly-both specifications as secondary evidence.

\subsection{Route B event study}

Let $M_{\ell}$ denote the set of SC Malaysia stock codes in the official compliant list for list date $\ell$. A stock code is an inclusion if it is in $M_{\ell}$ but not $M_{\ell-1}$, an exclusion if it is in $M_{\ell-1}$ but not $M_{\ell}$, and continuous compliant if it is in both $M_{\ell}$ and $M_{\ell-1}$. Events are then linked to Compustat securities using the stable crosswalk. Let $a_{\ell}$ denote the release date printed in the PDF for list date $\ell$. For event-study timing, $d=0$ is the first security-specific trading day on or after $a_{\ell}$.

For the Malaysia events, daily abnormal returns are market-adjusted returns, $AR_{i,d}=r_{i,d}-r^{MYS}_{m,d}$, where $r^{MYS}_{m,d}$ is the equal-weighted Compustat Global Malaysia daily return in the analysis universe. Event-window CARs sum market-adjusted returns over the window,
\begin{equation}
\CAR_{i,\ell}(\tau_1,\tau_2)=\sum_{d=\tau_1}^{\tau_2}AR_{i,d}.
\end{equation}
The Route A and Route B CARs both sum daily abnormal returns, but Route A uses a firm-specific CRSP market model whereas Route B uses a Malaysia market-adjusted benchmark. The two routes are not pooled in one regression. For each treated event, the candidate pool consists of continuously compliant securities from the same list date with complete pre-event log market capitalization, return, volatility, turnover, and Amihud illiquidity. The five variables are standardized within the treated observation and its date-specific candidate pool. We compute Euclidean distance and select the three nearest controls independently for each treated event, without a caliper and with replacement. Controls may therefore be reused within and across review cycles; missing candidates are excluded, and deterministic ties are resolved by event identifier. This is greedy nearest-neighbor matching rather than a globally optimal assignment. A robustness design retains a treated event only when at least three same-date controls share its two-digit SIC industry and does not fall back to another industry. For treated firm $i$ and matched controls $\mathcal{C}(i,\ell)$, the matched abnormal-return difference is
\begin{equation}
\diffCAR_{i,\ell}(\tau_1,\tau_2)
=
\CAR_{i,\ell}(\tau_1,\tau_2)
-
\frac{1}{|\mathcal{C}(i,\ell)|}\sum_{j\in\mathcal{C}(i,\ell)}
\CAR_{j,\ell}(\tau_1,\tau_2).
\end{equation}

Our primary Route B event-study specification focuses on continuously listed, liquidity-qualified SC Malaysia inclusion events. A continuously listed transition requires the security's first observed trading date to precede the immediately previous SC review; this removes new listings and security-history changes that enter the official universe between reviews. Liquidity is measured before the release date as average share turnover over available trading days in $[-30,-1]$. If $\mathcal{D}_{i,\ell}^{pre}$ is that set of available pre-event trading days, then
\begin{equation}
\overline{Turn}_{i,\ell}^{[-30,-1]}
=
\frac{1}{|\mathcal{D}_{i,\ell}^{pre}|}\sum_{d\in\mathcal{D}_{i,\ell}^{pre}}\frac{Volume_{i,d}}{SharesOut_{i,d}}.
\end{equation}
Let $Q^{turn}_{0.25}$ be the bottom-quartile cutoff of $\overline{Turn}_{i,\ell}^{[-30,-1]}$ in the full matched event universe before treatment-specific filtering. The primary liquidity-qualified inclusion indicator is
\begin{equation}
\mathcal{L}_{i,\ell}
=
\mathbf{1}\left\{\overline{Turn}_{i,\ell}^{[-30,-1]}\ge Q^{turn}_{0.25}\right\}.
\end{equation}
This screen is pre-event and economically motivated: a formal permission label should be more likely to move prices when the security is sufficiently tradable for mandate-constrained capital to adjust. The primary control group consists of three continuously compliant securities matched within the same SC list date on pre-event market capitalization, pre-event returns, volatility, turnover, and illiquidity. The primary outcome is the matched treatment-minus-control cumulative abnormal return over the $[0,10]$ trading-day window for inclusion events with $\mathcal{L}_{i,\ell}=1$. We use this window because official classification changes may take several trading days to be incorporated by mandate-constrained investors, intermediaries, and investment platforms, while remaining short enough to limit contamination from unrelated firm news. Inference clustered by SC list date, together with wild-cluster bootstrap $p$-values by list date, is treated as the primary inference. The full release-date-corrected inclusion sample is retained as the benchmark rather than discarded.

This creates an explicit evidence hierarchy. Pair-level tests describe the average treated-control contrast across matched pairs, but list-date clustered and wild-cluster tests are the inferential benchmark because SC lists are released semi-annually and many firm events share the same announcement date. The clustered tests have only 24 non-baseline list-date clusters, so they are intentionally conservative and are interpreted as such.

The design addresses the main identification threats directly. Same-list-date matching absorbs market-wide Malaysia news and list-date shocks common to treated and control firms; release-date event timing aligns day 0 with the public information date printed in the official PDF; matching on size, pre-event returns, volatility, turnover, and illiquidity reduces observable differences in return dynamics; and defining events at the official stock-code level before Compustat matching prevents security-linkage noise from becoming artificial treatment. A timing placebo shifts the same fixed matched sets three months before the public release, between official reviews. The older six-month shift is retained only as a preceding-review falsification because it coincides approximately with another official classification cycle. Date-clustered and wild-cluster inference address common list-date shocks, while two-way date-security clustering and a first-inclusion-only sample address repeated treated securities. The remaining risks are firm-specific news, partial anticipation, residual pre-event differences, reused controls, and only 24 list-date clusters. Route B is therefore interpreted as a conditional price association rather than a causal law.

All other event windows, exclusion-event estimates, unmatched comparisons, timing falsifications, and strict SIC-2 specifications are reported as consistency checks, asymmetry tests, robustness checks, or diagnostics rather than as separate primary tests.

\section{Route A Results: Classification Uncertainty Across Standards}\label{sec:us_results}

\subsection{Classification uncertainty ranks screen-implied transitions}

Table \ref{tab:reclass} reports the monitoring validation. The test is not causal: proximity to an implemented threshold is mechanically related to crossing that threshold, and disagreement partly reflects the same boundary geometry. The operational question is whether the variables rank next-month screen-implied eligibility instability. Implemented disagreement has a coefficient of 1.1609 conditional on proximity, while proximity has a coefficient of 3.2384. With two-way clustering by security and calendar month, the corresponding $z$-statistics are 39.75 and 30.61. Larger securities are less likely to experience an implemented transition. The regression uses a stratified random sample of 300,000 security-months; its purpose is ranking validation rather than inference about an economic cause of official provider decisions.

\input{table_routea_monitoring_clustered}

\begin{figure}[!htbp]
\centering
\includegraphics[width=0.78\textwidth]{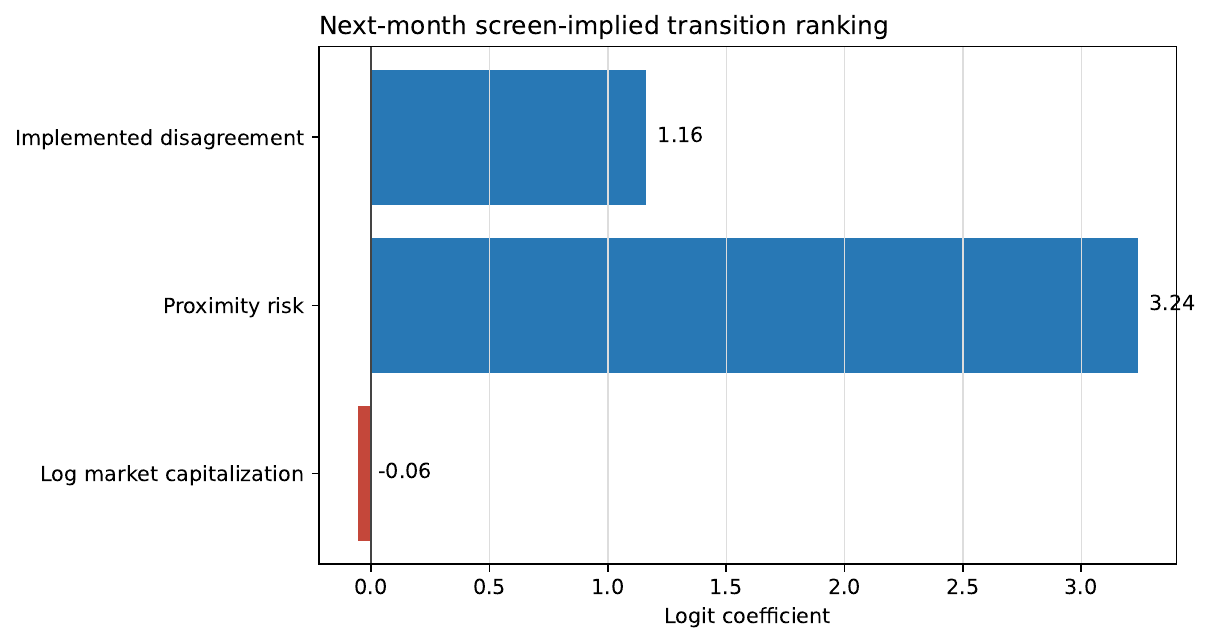}
\caption{Monitoring-state validation. Implemented disagreement and proximity to active screening boundaries rank next-month screen-implied eligibility transitions.}
\AltText{Bar chart of logit coefficients for next-month screen-implied transitions. Proximity risk has the largest positive coefficient, implemented disagreement is also positive, and log market capitalization is slightly negative.}
\label{fig:monitoring}
\end{figure}
\FloatBarrier

Out-of-sample AUC, Brier-score, and decile-lift diagnostics are reported in the Internet Appendix so the present paper does not turn the monitoring validation into a separate prediction study. They confirm ranking value, but they do not remove the mechanical connection between implemented boundaries and implemented transitions. The defensible conclusion is narrow: the measures identify holdings whose researcher-emulated investability is unstable at the next monitoring date.

The Internet Appendix also carries the positive-IDIT-only income proxy through the downstream Route A tests. Mean eligibility rises from 0.276 to 0.323 and mean disagreement from 0.225 to 0.289, so the measurement choice is economically material. Nevertheless, the rebuilt monitoring coefficients remain positive: 1.0687 for disagreement and 3.1006 for proximity under security-and-month clustering. The monitoring ranking is therefore not an artifact of including NOPI, although the levels of eligibility and disagreement are proxy-sensitive.

\subsection{Fama--MacBeth expected-return diagnostics}

Table \ref{tab:fmb_quality} reports Fama--MacBeth diagnostics. These regressions are a diagnostic contrast rather than a central pricing test. Eligibility share is positive in the baseline specification, opposite to the sign one would obtain by pretending that an equal rulebook count is structural permitted investor mass. Once controls include profitability, investment, capital expenditure, and the screening ratios, the eligibility-share coefficient is economically zero and statistically insignificant. Implemented disagreement remains positive after the expanded controls; this is a residual characteristic association, not evidence that disagreement is a priced state variable. The table therefore closes, rather than evades, the empirical distinction between equal rulebook coverage and mandate-weighted capital.

\input{table_fmb_quality_controls}

Under the IDIT-only proxy, the fully adjusted eligibility coefficient is $-0.0005$ ($t=-1.94$), while the disagreement coefficient remains positive at 0.0006 ($t=2.42$). Thus the absence of a robust positive eligibility premium survives, whereas the exact disagreement slope remains a proxy-dependent characteristic association. Route A validates a monitoring object and rejects an unconditional equal-weighted exclusion-premium interpretation; it does not structurally test Proposition~\ref{prop:mass}.

\subsection{September 2023 DJIM/S\&P methodology shock}

Table \ref{tab:daily_event_us} reports conventional summed daily market-model abnormal returns. We interpret the September 2023 DJIM/S\&P methodology shock as a boundary-condition test: a rulebook change that mechanically expands formal investability need not reprice firms unless it also expands marginal demand. The announcement-date results are positive in the broad unmatched newly-both specification: 0.750 percentage points over $[0,1]$ and 1.210 percentage points over $[0,3]$. The effects are not statistically significant after matching. Around the September 1 pro-forma date, matched estimates are small or negative. Figure \ref{fig:us_event} shows the pattern across windows.

\input{table_routea_daily_event_corrected}

\begin{figure}[!htbp]
\centering
\includegraphics[width=0.86\textwidth]{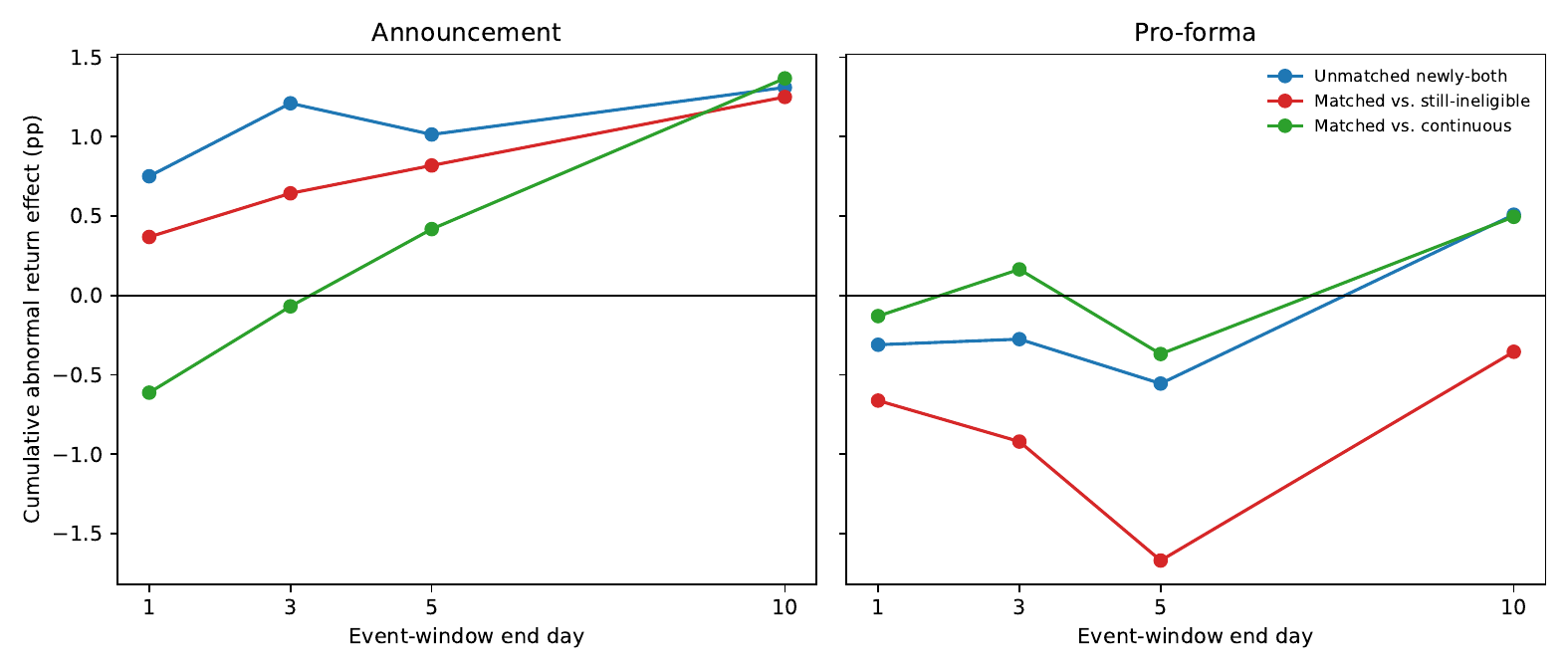}
\caption{September 2023 DJIM/S\&P methodology-shock estimates. Announcement-window effects are positive in unmatched samples but attenuate after matching; pro-forma matched estimates do not show robust positive repricing.}
\AltText{Grouped line chart of event-window CAR effects for the September 2023 DJIM/S\&P methodology change. Unmatched announcement effects are positive, matched announcement effects are smaller, and matched pro-forma effects are flat or negative.}
\label{fig:us_event}
\end{figure}

Placebo evidence reinforces caution. Table \ref{tab:us_placebo} reports comparable September dates in non-event years. Most windows are insignificant, but 2020 shows a positive $[0,3]$ placebo and 2022 shows strongly negative $[0,3]$ and $[0,5]$ effects. Figure \ref{fig:us_placebo} visualizes these placebo coefficients. The conclusion is not that the 2023 pro-forma effect is negative; it is that the price evidence is too fragile to support an unconditional inclusion-premium interpretation.

\input{table_routea_placebo_corrected}

\begin{figure}[!htbp]
\centering
\includegraphics[width=0.84\textwidth]{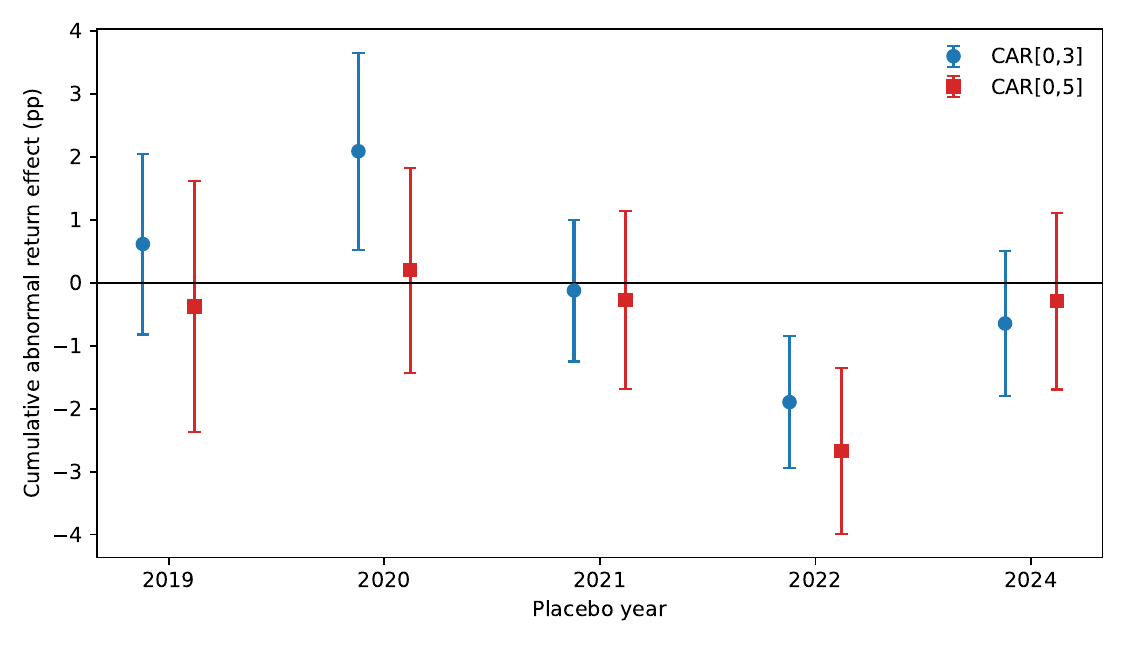}
\caption{U.S. September placebo tests. Points report coefficients and whiskers report 95\% confidence intervals for discrete placebo years; observations are not connected across years. The mixed pattern cautions against interpreting a single matched event-window estimate as mechanical proof of price pressure.}
\AltText{Grouped coefficient plot of placebo CAR effects for discrete September dates in non-event years, with separate points and confidence intervals for the zero-to-three-day and zero-to-five-day windows.}
\label{fig:us_placebo}
\end{figure}

\FloatBarrier
\section{Route B Results: Official SC Malaysia Classification Events}\label{sec:malaysia}

\subsection{Raw event returns}

Table \ref{tab:malaysia_raw} reports release-date-corrected raw market-adjusted CARs by event type. Inclusions are positive over the medium $[0,10]$ window, while exclusions are small and not persistently negative under release-date timing. Continuous-compliant firms have returns close to zero. These raw differences motivate the matched analysis but are not the main evidence because event firms differ in size, liquidity, and pre-event dynamics.

\begin{table}[!htbp]
\centering
\caption{Route B raw event CARs by classification event}
\label{tab:malaysia_raw}
\begin{tabular}{lrrrrr}
\toprule
Event type & $N$ & $[-10,-1]$ & $[0,3]$ & $[0,5]$ & $[0,10]$ \\
\midrule
Continuous compliant & 12,024 & -0.11\% & -0.09\% & -0.09\% & -0.28\% \\
Exclusion & 375 & 0.46\% & 0.02\% & -0.09\% & -0.62\% \\
Inclusion & 595 & -0.85\% & -0.32\% & 0.00\% & 1.14\% \\
\bottomrule
\end{tabular}
\end{table}

\subsection{Matched release-date-corrected inclusion CARs}

Table~\ref{tab:routeb_matching_balance} reports matched-pair balance for the corrected primary sample. The treated and matched-control securities are close on log market capitalization, pre-event returns, volatility, turnover, and Amihud illiquidity. The largest standardized difference is 0.150 for log dollar volume; all matching variables are below 0.09 in absolute value.

\input{table_routeb_matching_balance}
\FloatBarrier

Table \ref{tab:routeb_fundamentals} reports the complete specification ladder. All 582 matched inclusions, including first appearances after the preceding review, produce 0.78 percentage points over $[0,10]$ and 1.44 points over $[0,20]$. Removing listing-history contamination leaves 410 continuously listed inclusions. Their estimates are 0.90 and 1.46 percentage points, with wild-cluster $p$-values above conventional significance levels. The turnover floor then leaves 295 primary events with 1.76 percentage points over $[0,10]$ and 2.25 points over $[0,20]$. The table therefore makes explicit that the stronger inference emerges at the turnover-screen step, not merely from removing new listings.

\input{table_routeb_fundamentals}

The Compustat Global fundamentals merge originally produced no usable events because numeric-looking GVKEYs were read with a trailing ``.0'' in the event file. After normalizing both sources to the same six-character identifier, 834 matched events have two pre-event fiscal observations. The 262 primary inclusions with complete treated and control histories have unadjusted estimates of 2.15 and 2.75 percentage points. Adding treated-minus-matched-control changes in debt, cash, receivables, book leverage, log assets, and sales-to-assets, together with market controls, leaves the centered intercepts at those same common-sample means while reducing their standard errors. The larger coefficients relative to the 295-event row therefore reflect complete-case sample composition, not the addition of controls. The cached panel does not contain usable ROA, and the paper does not imply otherwise.

\subsection{Event-time profile and pre-event diagnostics}

Table~\ref{tab:routeb_event_windows} reports all standard windows for the corrected primary sample. The response is not immediate: $[0,1]$ is negative and insignificant, $[0,3]$ is positive but only marginal under wild inference, and the estimate accumulates through days 10 and 20. The cumulative pre-event contrasts over $[-20,-1]$, $[-10,-1]$, and $[-5,-1]$ are all statistically insignificant.

\input{table_routeb_event_windows}

Figure~\ref{fig:malaysia_dynamic} plots the matched daily path normalized to zero at day $-1$. A list-date-clustered joint Wald test of the 20 daily pre-event coefficients rejects equality to zero ($p=0.006$), driven by isolated differences around days $-17$ and $-11$ rather than a monotonic run-up in the immediate pre-event window. With 20 restrictions and only 24 release-date clusters, that test is demanding and potentially size-sensitive, but it cannot be ignored. The combination of null cumulative pre-windows and a rejected long daily joint test limits a clean parallel-trends or causal interpretation.

\begin{figure}[!htbp]
\centering
\includegraphics[width=0.88\textwidth]{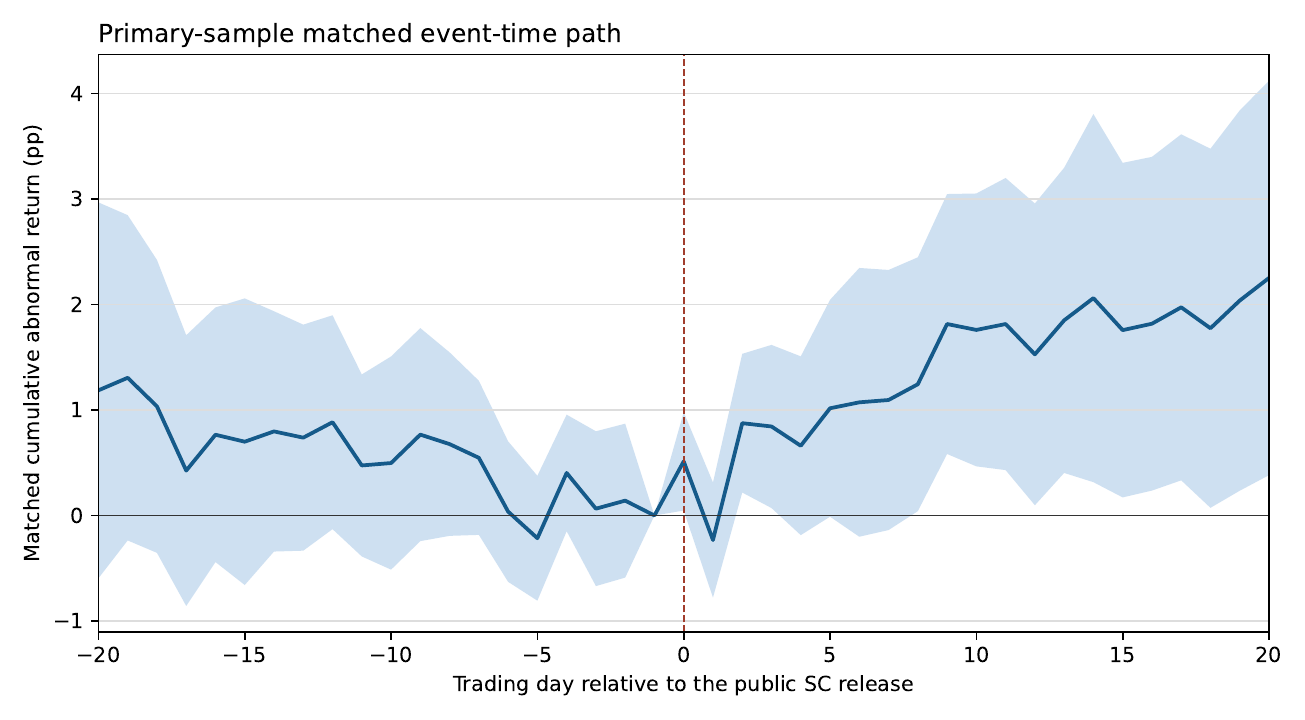}
\caption{Matched event-time path for the primary SC Malaysia inclusion sample. The path is normalized to zero at day $-1$; shaded bands use list-date-clustered standard errors.}
\AltText{Matched cumulative abnormal-return path from 20 trading days before through 20 days after the public SC Malaysia release. The pre-event path is not monotonic, and the post-event effect accumulates mainly after day 3.}
\label{fig:malaysia_dynamic}
\end{figure}
\FloatBarrier

\subsection{What the turnover screen does and does not establish}

Table~\ref{tab:malaysia_liquidity_diagnostics} recomputes nearby pre-event screens after excluding new listings. The $[0,10]$ estimate remains positive under the reclassification-sample turnover quartile, tercile, and median cutoffs, with wild-cluster $p$-values of 0.014, 0.038, and 0.009. Market-cap screens do not reproduce the result. This reduces the concern that the fixed cutoff is an isolated numerical choice, but it does not make the result unconditional.

\input{table_routeb_threshold_sensitivity}

The primary turnover floor is a sample definition based exclusively on pre-event information, not evidence of a monotonic liquidity mechanism. Table~\ref{tab:routeb_turnover_slopes} makes this distinction explicit by parameterizing separate slopes for continuously listed inclusions and exclusions. The inclusion slope is 0.09 percentage points per standard deviation over $[0,10]$ ($p_{\mathrm{wild}}=0.925$) and $-0.03$ over $[0,20]$ ($p_{\mathrm{wild}}=0.979$). The previously reported positive inclusion-by-turnover interaction came from a strongly negative exclusion slope. It cannot be interpreted as the inclusion effect rising with turnover. The defensible result is therefore conditional on a pre-event turnover floor; the data do not establish a linear dose-response relation within included securities.

\input{table_routeb_turnover_slopes}
\FloatBarrier

\subsection{List-date influence}

The small number of SC Malaysia list-date clusters is the main inference constraint. Figure~\ref{fig:malaysia_jackknife} therefore reports leave-one-list-date-out estimates for the corrected primary sample. The $[0,10]$ estimate remains positive when every release is omitted in turn, ranging from 1.29 to 1.98 percentage points. The $[0,20]$ range is 1.78 to 2.55 percentage points. The level is therefore not driven by one release, although 24 clusters remain a material limitation on inference.

\begin{figure}[!htbp]
\centering
\includegraphics[width=0.86\textwidth]{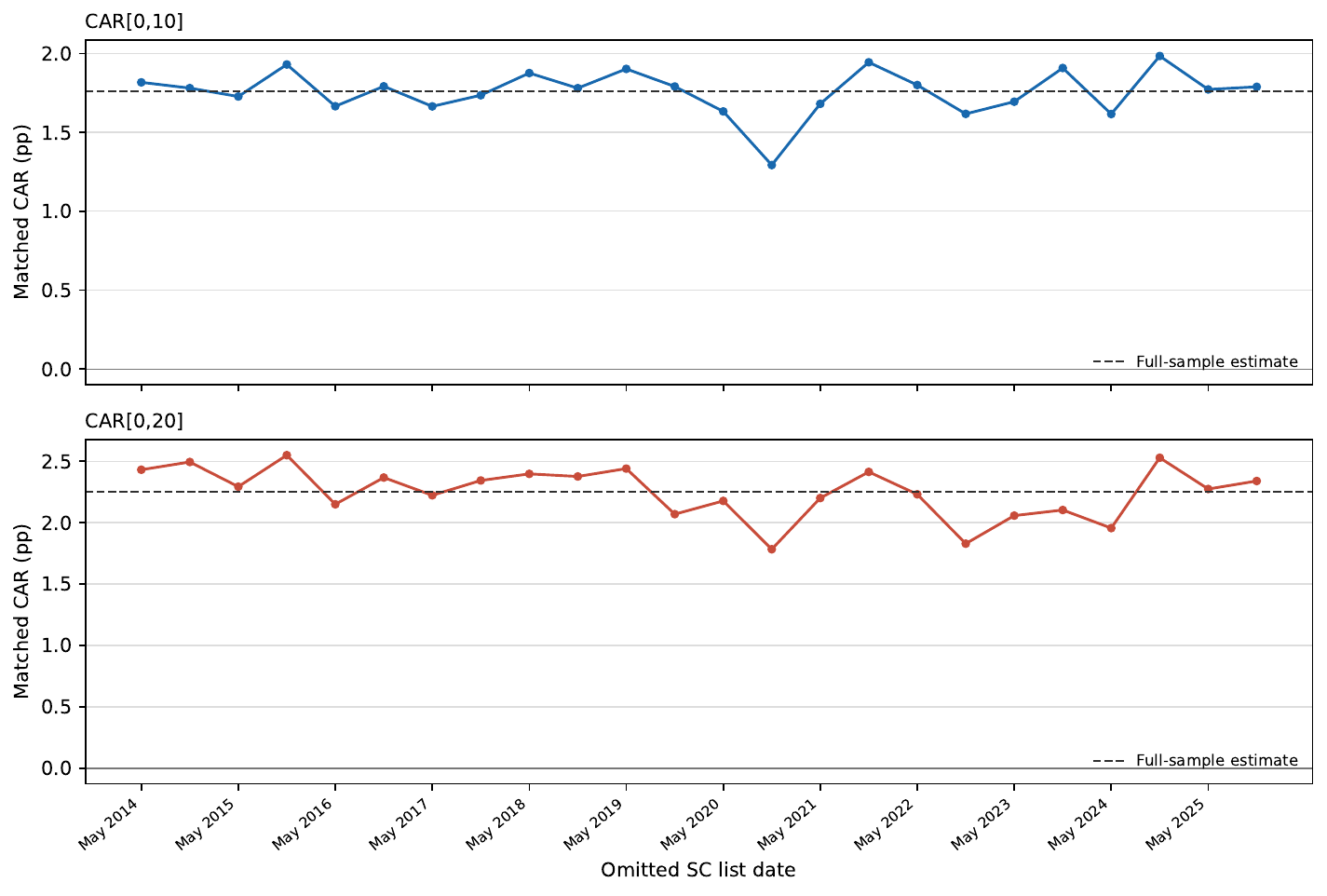}
\caption{Leave-one-list-date-out estimates for continuously listed, liquidity-qualified SC Malaysia inclusions. Dashed lines report the corresponding corrected full-sample estimates for each event window.}
\AltText{Line chart showing the matched inclusion diffCAR when each SC Malaysia list date is omitted in turn. All leave-one-date-out estimates remain positive and close to the full-sample estimate.}
\label{fig:malaysia_jackknife}
\end{figure}
\FloatBarrier

\subsection{Industry and repeated-security robustness}

The primary sample contains 295 events for 237 distinct treated securities, so 44 securities appear more than once. Two-way clustering by list date and treated security yields $p$-values of 0.008 and 0.015 over $[0,10]$ and $[0,20]$. Retaining only the first primary inclusion per security leaves 237 observations and estimates of 1.69 and 2.41 percentage points, both significant under list-date wild-cluster inference. The Internet Appendix reports these results and audits control reuse. A stricter same-date, same-SIC design retains 266 treated events and produces positive but statistically imprecise estimates of 0.87 and 1.48 percentage points. Thus repeated treated securities do not generate the primary estimate, while exact industry comparability lowers precision and effect size.

\subsection{Historical ownership diagnostics}

The event-study evidence is strongest when it is interpreted as a conditional price-of-permission result: the official label matters when it can coordinate investors who are permitted, willing, and able to trade. To examine this mechanism directly, we reconstruct dated LSEG holder snapshots around each SC Malaysia list cycle for the RIC-mapped treated-event universe. For each event, we observe investor names and holder types approximately one month before and one month after the official list date, classify holders into Shariah-sensitive, Malaysian state/pension, sovereign/Gulf, broad mandate-constrained, and global passive buckets, and compute pre-to-post changes in the number of investors and percentage ownership. We also create portfolio-only constrained buckets that exclude likely strategic or control holders.

Table~\ref{tab:routeb_ownership_did} reports absolute ownership changes and the matched-control difference-in-differences comparison requested by the mechanism claim. Shariah-sensitive portfolio ownership rises by 0.022 percentage points among the 288 treated inclusions with dated coverage and is unchanged among their controls. The treated change is borderline under wild-cluster inference ($p_{\mathrm{wild}}=0.051$), but the treated-minus-control difference is imprecise ($p_{\mathrm{wild}}=0.360$). Broad constrained ownership produces a 0.053-percentage-point difference with $p_{\mathrm{wild}}=0.282$; the global passive placebo is null. Thus the snapshots show suggestive mandate-sensitive reallocation, not an identified constrained-buyer first stage relative to matched controls.

\input{table_routeb_ownership_did}

\subsection{Channel-narrowing diagnostics}

The Internet Appendix reports corrected-sample demand-pressure, abnormal-volume, and ETF-capacity diagnostics. In the 288 primary events with the required LSEG variables, the predicted constrained-pressure proxy is positively related to subsequent constrained ownership change ($p_{\mathrm{wild}}=0.037$), but its CAR coefficient is negative and insignificant. The pressure-by-turnover term is also insignificant. Realized buying pressure is a post-event, endogenous variable; its CAR coefficient is negative and insignificant, while its abnormal-volume association is only marginal. Listed ETF capacity and share creations likewise do not explain CARs. Direct pressure proxies are therefore either insignificant or enter with signs inconsistent with a simple buying-pressure account. These results narrow the channel but do not establish a broader mechanism by elimination.

\subsection{Timing falsifications and exclusion-window interpretation}

Table~\ref{tab:malaysia_placebo} shifts each primary event and its fixed matched set three calendar months before the public release, placing the pseudo-event between official review dates. Estimates from $[0,1]$ through $[0,20]$ are positive but small and statistically indistinguishable from zero. This is a genuine non-review timing falsification. It does not repair the rejected long daily pretrend test, but it shows that the primary matched sets do not reproduce the post-release magnitude at a nearby midpoint date.

\input{table_routeb_placebo_clean}

The former six-month shift coincides approximately with the preceding official review. It is retained in the Internet Appendix and relabelled a preceding-review falsification rather than a non-event placebo. Under release-date timing, actual exclusion estimates are not persistently negative and are not robust under list-date clustered or wild inference. Exclusions therefore remain asymmetry diagnostics rather than a separate pricing result.

\section{Economic Interpretation and Portfolio Implications}\label{sec:interpretation}

The evidence yields three linked conclusions about Shariah investability. First, a current binary pass label omits a dynamic risk dimension. Cross-standard disagreement identifies fragmentation in the current eligible investor set, while proximity to active financial boundaries identifies securities whose researcher-emulated classifications are less stable. Both variables rank next-month screen-implied transitions in the U.S. panel. That monitoring value is relevant for mandate governance and portfolio turnover, but it is partly mechanical and is not evidence that the variables cause official-provider decisions.

Second, formal Shariah eligibility is not equivalent to a priced expansion in the investor base. The U.S. Fama--MacBeth diagnostics do not support an unconditional permission premium in equal-weighted rulebook data; profitability, investment, and screening-ratio controls absorb the eligibility-share slope. The September 2023 DJIM/S\&P methodology shock also produces no robust matched repricing, and its placebo pattern is unstable. These findings do not reject the permitted-investor-mass comparative static because historical mandate-capital weights are not observed. They reject the simpler substitution of rulebook counts for economically marginal constrained capital.

Third, official SC Malaysia inclusions exhibit a conditional positive return pattern. The 410 continuously listed inclusions are positive but imprecise, while conventional significance emerges after applying the pre-event turnover floor. The 295-event estimate survives leave-one-date-out, first-inclusion-only, two-way date-security clustering, and a null mid-review timing placebo. Its larger complete-case coefficient is already present before controls are added. Nearby turnover thresholds reproduce the result, whereas market-cap restrictions do not. These checks make listing contamination, one release, and a unique cutoff less plausible explanations.

The same evidence prevents a stronger causal interpretation. The inclusion return does not increase monotonically with turnover, strict same-SIC estimates are imprecise, and the joint test of 20 daily pre-event coefficients rejects. Institutional recognition distinguishes the SC Malaysia setting conceptually, but it is not a separately identified treatment. Predicted pressure forecasts some constrained ownership adjustment but not CARs, and matched ownership changes remain imprecise. The paper therefore identifies neither a security-level demand instrument nor a unique marginal buyer.

The U.S.--Malaysia contrast is consequently an institutional comparison rather than a comparison of one failed and one successful event study. The U.S. evidence shows that formal rulebook eligibility can change without robust matched repricing. The Malaysia evidence shows that an official classification can have price content in a recognized local Shariah market when the security also satisfies an ex ante tradability condition. Country, market structure, event type, and investor composition differ simultaneously, so the comparison does not causally identify recognition. It instead establishes the paper's central boundary condition: formal permission alone is insufficient.

For Shariah-constrained portfolio managers, the practical implication is to separate current eligibility, transition risk, and realized event risk. Current eligibility determines whether a position may be held. Cross-standard disagreement and boundary distance identify holdings whose permission is fragmented or comparatively fragile. Authority-specific mandate exposure and tradability determine whether an actual reclassification is likely to require economically meaningful adjustment. Used jointly, these objects can support surveillance of possible forced-divestment exposure, candidate entrants, expected turnover, and liquidity needs around official reviews. Treating them separately also prevents every near-boundary firm from being labelled an imminent reclassification and every official event from being treated as a mechanical trading signal.

The paper does not estimate the profitability of acting on this information. A tradable early-warning strategy would require out-of-sample forecasts of official rather than researcher-emulated classifications, a decision rule fixed before each review, implementation lags, transaction costs, and investable capacity. Those are distinct empirical requirements. The present contribution is the prior step: it identifies dynamic investability as portfolio-relevant information and tests when realized permission changes have price content. The broader finance relevance follows from the same investor-set logic. Green taxonomies, exclusionary mandates, and sustainability benchmarks can also segment capital when recognized institutional classifications define feasible investment universes.

Three limitations bound the contribution. Historical mandate-capital weights are unavailable, so the structural $m_i$ is not directly estimated. The LSEG ownership, demand-pressure, abnormal-volume, and ETF diagnostics do not isolate a price-setting shock. Finally, the rejected long daily pretrend test and the dependence of conventional significance on the turnover-qualified sample limit causal interpretation. These constraints are substantive. The defensible empirical claim is a conditional price association for official Shariah permission, not causally isolated mandate-constrained buying.

\section{Conclusion}\label{sec:conclusion}

Shariah investability is not only a binary current state. Recognized screening standards can disagree over the same firm, and firms with identical current labels can lie at very different distances from the boundaries that determine future eligibility. Cross-standard disagreement and boundary proximity therefore reveal complementary dimensions of classification risk: fragmentation of the current feasible investor set and instability of that set over time.

The permitted-investor-mass framework identifies the economic object behind those classifications. A Shariah label matters through the constrained capital governed by the relevant standard or authority, not through a simple count of approving rulebooks. Route A supports the monitoring implication but not an unconditional pricing interpretation. Disagreement and proximity rank next-month screen-implied transitions, while equal-weighted eligibility and the September 2023 DJIM/S\&P methodology event do not provide robust evidence of a permission premium.

Official SC Malaysia classifications provide the paper's central price evidence. The corrected primary sample of 295 continuously listed, liquidity-qualified inclusions earns 1.76 percentage points over $[0,10]$ and 2.25 points over $[0,20]$ under list-date clustered and wild-cluster inference. The estimates survive leave-one-date-out, first-inclusion-only, date-security clustering, and relative-fundamentals specifications, and the mid-review timing placebo is null. The result nevertheless remains conditional: the 410-event continuously listed sample is imprecise, the inclusion effect does not rise linearly with turnover, strict industry matching weakens precision, and the long daily pre-event joint test rejects. Ownership diagnostics do not identify a unique marginal buyer.

Formal Shariah classification is therefore neither irrelevant nor mechanically priced. It can have price content when it defines an institutionally recognized investment universe and the affected security can be traded, but the evidence does not isolate a clean causal demand shock. The portfolio interpretation is precise: a current label tells a constrained manager what may be held, while classification-boundary risk identifies where that feasible universe is vulnerable to change. This information can support mandate-risk monitoring and preparation for rebalancing, but the paper does not establish a profitable pre-reclassification strategy. The broader finance implication is that institutional classifications matter when they alter feasible investor participation. Applications to other constrained investment universes follow from that investor-base mechanism; the central contribution remains the conversion of static Shariah eligibility into a dynamic, financially relevant classification-risk state.

\clearpage

\section*{Declaration of competing interests}

The authors declare that they have no known competing financial interests or personal relationships that could have appeared to influence the work reported in this paper.

\section*{Funding}

This research did not receive any specific grant from funding agencies in the public, commercial, or not-for-profit sectors.

\section*{Data availability}
The replication package for this paper will include the programs needed to reproduce the empirical tables and figures, the data dictionary, merge-key documentation, parser logs, and synthetic input files that illustrate the workflow. Public Securities Commission Malaysia list files are cited in the reference list and can be obtained from the Securities Commission Malaysia Islamic Capital Market publication archive. CRSP, Compustat, WRDS, LSEG Workspace, and index-provider inputs are proprietary and cannot be redistributed. Users with appropriate licenses can reproduce the restricted-data components using the supplied scripts. Where restricted source data cannot be shared, the replication package will include pseudo data, derived non-proprietary outputs where permitted, and run logs generated from the licensed-data workflow.

\clearpage
\section*{Declaration of generative AI use}

During the preparation of this work, the authors used Claude (Anthropic) and Codex (OpenAI) to solely assist with LaTeX debugging. All formulas, empirical claims, references, notation, and manuscript content were subsequently reviewed and verified by the authors, who take full responsibility for the content of the submitted manuscript.

\section*{Acknowledgements}

The authors thank colleagues and seminar participants for helpful comments and discussions. We are particularly grateful for feedback on the paper's theoretical framework, empirical design, data construction, and event-study interpretation. The paper benefited from discussions on investor segmentation, classification uncertainty, and mandate-constrained capital. The authors acknowledge the support of the Institute of Finance and Technology at University College London. All remaining errors are our own.

\bibliographystyle{elsarticle-harv}
\bibliography{references}

\clearpage
\input{appendix_content}

\end{document}

%% file: table_routea_weight_sensitivity.tex
\begin{table}[!htbp]
\centering
\caption{Mandate-weight scenario sensitivity for the permitted-investor-mass proxy}
\label{tab:weight_sensitivity}
\begin{adjustbox}{max width=\textwidth}
\begin{tabular}{lL{3.7cm}cccc}
\toprule
Scenario & Weight tilt & Spearman $\rho$ vs. equal & Mean rank shift & Top-quintile overlap & Bottom-quintile overlap \\
\midrule
        Equal rulebook weights & All seven standards $=1/7$ & 1.000 & 0.0\% & 100.0\% & 100.0\% \\
        DJIM/S\&P-heavy & DJIM $=0.30$, S\&P $=0.30$ & 0.996 & 0.8\% & 89.4\% & 100.0\% \\
        MSCI-heavy & MSCI main $=0.25$, M-Series $=0.25$ & 1.000 & 0.1\% & 99.0\% & 100.0\% \\
        SC-heavy & SC Malaysia $=0.40$ & 0.992 & 1.3\% & 87.8\% & 100.0\% \\
        AAOIFI-heavy & AAOIFI $=0.40$ & 0.995 & 0.9\% & 91.9\% & 100.0\% \\
\bottomrule
\end{tabular}
\end{adjustbox}
\begin{minipage}{0.96\textwidth}
\footnotesize \emph{Notes:} The table recomputes the permission score $\sum_s \omega_s e_i^{(s)}$ on the full Route A firm-month panel under transparent mandate-weight scenarios. The tilted scenarios are deliberately severe stress tests that concentrate mandate weight in selected rulebook families; they are not estimated capital shares. Rank shifts and quintile overlaps are computed within calendar month. The exercise is a calibration/sensitivity check, not an estimate of historical mandate-weighted capital.
\end{minipage}
\end{table}

%% file: table_routeb_sample_flow.tex
\begin{table}[!htbp]
\centering
\caption{Route B event-sample reconciliation}
\label{tab:routeb_sample_flow}
\begin{threeparttable}
\begin{tabular}{lrr}
\toprule
Panel A: cumulative event-sample stage & Inclusions & Exclusions \\
\midrule
Official stock-code transitions & 859 & 661 \\
Linked to a Compustat Global security & 598 & 429 \\
Valid release-date CAR & 595 & 375 \\
Successful same-date three-control match & 582 & 367 \\
Continuously listed before the preceding review & 410 & 367 \\
Continuously listed and above turnover floor & 295 & 257 \\
\midrule
Panel B: non-cumulative fundamentals coverage & Inclusions & Exclusions \\
\midrule
Matched events with two prior fiscal observations & 507 & 327 \\
Primary inclusions with complete relative controls & 262 & -- \\
\bottomrule
\end{tabular}
\begin{tablenotes}[flushleft]
\footnotesize
\item Panel A follows cumulative event selection. Panel B is deliberately separate because two-fiscal-observation coverage is measured among all matched linked events and is not a stage after the continuously listed restriction. Events are created from consecutive official stock-code lists before security matching. The valid-CAR row reports the 595 inclusions and 375 exclusions used in the raw event-return table.
\end{tablenotes}
\end{threeparttable}
\end{table}

%% file: table_routea_monitoring_clustered.tex
\begin{table}[!htbp]
\centering
\caption{Classification uncertainty and next-month screen-implied transitions}
\label{tab:reclass}
\begin{threeparttable}
\begin{tabular}{lc}
\toprule
 & Next-month eligibility transition \\
\midrule
Implemented disagreement & 1.1609*** \\
 & (39.75) \\
Proximity risk & 3.2384*** \\
 & (30.61) \\
Log market capitalization & -0.0557*** \\
 & (-8.48) \\
Year fixed effects & Yes \\
Observations & 300,000 \\
\bottomrule
\end{tabular}
\begin{tablenotes}[flushleft]
\footnotesize
\item The dependent variable equals one when any researcher-emulated standard-specific label changes in the next consecutive calendar month. Entries are logit coefficients with $z$-statistics in parentheses. Standard errors are two-way clustered by CRSP security and calendar month. The exercise validates ranking information; it is not a causal test because the outcome and proximity measure share the implemented screening boundaries.
\end{tablenotes}
\end{threeparttable}
\end{table}

%% file: table_fmb_quality_controls.tex
\begin{table}[!htbp]
\centering
\caption{Fama--MacBeth permission slopes with quality controls}
\label{tab:fmb_quality}
\begin{threeparttable}
\begin{tabular}{llrr}
\toprule
Main variable & Controls & Mean coefficient & $t$-statistic \\
\midrule
Eligibility share & Baseline & 0.0013*** & 4.13 \\
Eligibility share & + quality & 0.0005 & 1.62 \\
Eligibility share & + quality + screens & -0.0001 & -0.35 \\
Implemented disagreement & Baseline & 0.0015*** & 4.47 \\
Implemented disagreement & + quality & 0.0010*** & 2.76 \\
Implemented disagreement & + quality + screens & 0.0007** & 2.54 \\
\bottomrule
\end{tabular}
\begin{tablenotes}[flushleft]
\footnotesize
\item The table reports time-series average Fama--MacBeth slopes for standardized permission variables. Baseline controls are log market capitalization, log book-to-market, 12-to-2 momentum, and one-month reversal. Quality controls add operating income over assets, ROA, asset growth, and capital expenditure over lagged assets. Screen-ratio controls add leverage, cash, receivables, and impure-income ratios used in the screening data. Standard errors are Newey--West with 12 lags. ***, **, and * denote significance at 1\%, 5\%, and 10\%.
\end{tablenotes}
\end{threeparttable}
\end{table}

%% file: table_routea_daily_event_corrected.tex
\begin{table}[!htbp]
\centering
\caption{Route A daily event-window cumulative abnormal returns}
\label{tab:daily_event_us}
\begin{adjustbox}{max width=\textwidth}
\begin{tabular}{llrrrr}
\toprule
Event date & Specification & CAR$[0,1]$ & CAR$[0,3]$ & CAR$[0,5]$ & CAR$[0,10]$ \\
\midrule
Announcement Aug. 4, 2023 & Unmatched newly-both & 0.750** & 1.210*** & 1.013* & 1.310* \\
 & & (2.31) & (2.63) & (1.79) & (1.85) \\
  & Matched vs. still-ineligible & 0.368 & 0.643 & 0.818 & 1.249 \\
 & & (0.89) & (1.09) & (1.04) & (1.37) \\
  & Matched vs. continuous & -0.612 & -0.069 & 0.418 & 1.367 \\
 & & (-1.04) & (-0.08) & (0.40) & (1.15) \\
\addlinespace
Pro-forma Sept. 1, 2023 & Unmatched newly-both & -0.311 & -0.275 & -0.555 & 0.508 \\
 & & (-1.02) & (-0.66) & (-1.08) & (0.78) \\
  & Matched vs. still-ineligible & -0.662 & -0.920 & -1.669** & -0.355 \\
 & & (-1.40) & (-1.46) & (-2.24) & (-0.37) \\
  & Matched vs. continuous & -0.130 & 0.164 & -0.369 & 0.495 \\
 & & (-0.26) & (0.25) & (-0.45) & (0.45) \\
\bottomrule
\end{tabular}
\end{adjustbox}
\begin{flushleft}
\footnotesize CARs are sums of daily market-model abnormal returns, winsorized at the 1st and 99th percentiles within each event-window cross-section. Reported effects are percentage-point coefficients with HC3 $t$-statistics.
\end{flushleft}
\end{table}

%% file: table_routea_placebo_corrected.tex
\begin{table}[!htbp]
\centering
\caption{Route A placebo daily event tests}
\label{tab:us_placebo}
\begin{tabular}{lrr}
\toprule
Placebo date & CAR$[0,3]$ & CAR$[0,5]$ \\
\midrule
Sept. 1, 2019 & 0.615 & -0.376 \\
 & (0.84) & (-0.37) \\
Sept. 1, 2020 & 2.088*** & 0.202 \\
 & (2.62) & (0.24) \\
Sept. 1, 2021 & -0.123 & -0.273 \\
 & (-0.21) & (-0.38) \\
Sept. 1, 2022 & -1.895*** & -2.670*** \\
 & (-3.53) & (-3.99) \\
Sept. 1, 2024 & -0.647 & -0.292 \\
 & (-1.11) & (-0.41) \\
\bottomrule
\end{tabular}
\begin{flushleft}
\footnotesize Placebos use the matched newly-both versus still-ineligible design. CARs are sums of daily market-model abnormal returns. HC3 $t$-statistics are in parentheses.
\end{flushleft}
\end{table}

%% file: table_routeb_matching_balance.tex
\begin{table}[!htbp]
\centering
\caption{Matching balance for the primary SC Malaysia reclassification sample}
\label{tab:routeb_matching_balance}
\begin{adjustbox}{max width=\textwidth}
\begin{tabular}{lrrr}
\toprule
Pre-event variable & Treated mean & Matched-control mean & Standardized difference \\
\midrule
Log market capitalization & 19.122 & 19.125 & -0.002 \\
Pre-event return & 0.001 & 0.001 & -0.004 \\
Pre-event volatility & 0.036 & 0.034 & 0.084 \\
Pre-event turnover & 0.0066 & 0.0058 & 0.065 \\
Pre-event Amihud & 7.99e-06 & 5.25e-06 & 0.085 \\
Log dollar volume & 12.992 & 12.689 & 0.150 \\
\bottomrule
\end{tabular}
\end{adjustbox}
\begin{flushleft}
\footnotesize Balance is computed across the three matched pairs for each of the 295 continuously listed, liquidity-qualified inclusions. Controls are continuously compliant securities from the same SC list date. Standardized differences use the pooled treated-control standard deviation.
\end{flushleft}
\end{table}

%% file: table_routeb_fundamentals.tex
\begin{table}[!htbp]
\centering
\caption{SC Malaysia inclusion specification ladder}
\label{tab:routeb_fundamentals}
\scriptsize
\begin{threeparttable}
\begin{adjustbox}{max width=\textwidth}
\begin{tabular}{lrrrrrrr}
\toprule
 & & \multicolumn{3}{c}{CAR$[0,10]$} & \multicolumn{3}{c}{CAR$[0,20]$} \\
\cmidrule(lr){3-5}\cmidrule(lr){6-8}
Sample/specification & $N$ & Estimate & $p_{date}$ & $p_{wild}$ & Estimate & $p_{date}$ & $p_{wild}$ \\
\midrule
All matched inclusions & 582 & 0.784 & 0.163 & 0.163 & 1.436* & 0.072 & 0.083 \\
Continuously listed only & 410 & 0.896 & 0.127 & 0.134 & 1.458* & 0.053 & 0.064 \\
Continuously listed + turnover floor & 295 & 1.759*** & 0.008 & 0.017 & 2.252** & 0.018 & 0.035 \\
Complete fundamentals sample, unadjusted & 262 & 2.154*** & 0.003 & 0.009 & 2.753*** & 0.005 & 0.012 \\
Complete sample + relative fundamentals & 262 & 2.154*** & 0.001 & 0.006 & 2.753*** & 0.003 & 0.017 \\
\bottomrule
\end{tabular}
\end{adjustbox}
\begin{tablenotes}[flushleft]
\footnotesize
\item Estimates are percentage-point means of treated-minus-three-control market-adjusted CARs. The continuously listed restriction requires trading before the preceding SC review. The turnover floor is fixed at 0.000296 from the pre-event matched universe. The complete sample requires all six treated and matched-control-relative fundamental changes and four market controls. The final row adds standardized treated-minus-control changes in debt, cash, receivables, book leverage, log assets, and sales-to-assets plus pre-event size, return, volatility, and Amihud illiquidity. Because controls are centered, its intercept is the common-sample mean; the comparison does not attribute the larger complete-sample mean to adding controls. Inference clusters by SC list date; $p_{wild}$ uses 999 restricted Rademacher replications.
\end{tablenotes}
\end{threeparttable}
\end{table}

%% file: table_routeb_event_windows.tex
\begin{table}[!htbp]
\centering
\caption{Matched event-window profile for the primary Malaysia sample}
\label{tab:routeb_event_windows}
\begin{threeparttable}
\begin{tabular}{lrrrr}
\toprule
Window & Estimate & Cluster SE & $p_{date}$ & $p_{wild}$ \\
\midrule
$[-20,-1]$ & -0.816 & 0.987 & 0.409 & 0.451 \\
$[-10,-1]$ & -0.475 & 0.440 & 0.281 & 0.298 \\
$[-5,-1]$ & -0.037 & 0.340 & 0.914 & 0.909 \\
\midrule
$[0,1]$ & -0.233 & 0.278 & 0.404 & 0.461 \\
$[0,3]$ & 0.844** & 0.395 & 0.032 & 0.068 \\
$[0,5]$ & 1.017* & 0.524 & 0.053 & 0.082 \\
$[0,10]$ & 1.759*** & 0.660 & 0.008 & 0.017 \\
$[0,20]$ & 2.252** & 0.954 & 0.018 & 0.035 \\
\bottomrule
\end{tabular}
\begin{tablenotes}[flushleft]
\footnotesize
\item The sample contains 295 continuously listed, liquidity-qualified inclusions. Estimates are percentage-point treated-minus-control CARs. Pre-event windows are disclosed to assess differential drift; post-event windows show that the response is not immediate and accumulates mainly after day 3. Inference clusters by list date and uses 999 wild-cluster replications.
\end{tablenotes}
\end{threeparttable}
\end{table}

%% file: table_routeb_threshold_sensitivity.tex
\begin{table}[!htbp]
\centering
\caption{Pre-event screen sensitivity in continuously listed inclusions}
\label{tab:malaysia_liquidity_diagnostics}
\begin{adjustbox}{max width=\textwidth}
\begin{tabular}{lrrrr}
\toprule
Sample definition & Treated & Mean diff. & $p_{date}$ & $p_{wild}$ \\
\midrule
Fixed primary turnover floor & 295 & 1.76\% & 0.008 & 0.017 \\
Turnover above reclassification-sample quartile & 307 & 1.84\% & 0.004 & 0.014 \\
Turnover above reclassification-sample tercile & 273 & 1.62\% & 0.022 & 0.038 \\
Turnover above reclassification-sample median & 205 & 2.53\% & 0.003 & 0.009 \\
Market cap above reclassification-sample quartile & 307 & 0.87\% & 0.173 & 0.174 \\
Market cap above reclassification-sample median & 205 & 0.60\% & 0.471 & 0.493 \\
\bottomrule
\end{tabular}
\end{adjustbox}
\begin{flushleft}
\footnotesize All samples exclude securities first observed after the preceding SC review. Cutoffs use only pre-release variables. Means are percentage-point matched treatment-minus-control CAR$[0,10]$ estimates. Inference clusters by SC list date; wild $p$-values use 999 restricted Rademacher replications.
\end{flushleft}
\end{table}

%% file: table_routeb_turnover_slopes.tex
\begin{table}[!htbp]
\centering
\caption{Turnover slopes in continuously listed classification events}
\label{tab:routeb_turnover_slopes}
\begin{threeparttable}
\begin{tabular}{lrrrr}
\toprule
 & \multicolumn{2}{c}{CAR$[0,10]$} & \multicolumn{2}{c}{CAR$[0,20]$} \\
\cmidrule(lr){2-3}\cmidrule(lr){4-5}
Term & Coefficient & $p_{wild}$ & Coefficient & $p_{wild}$ \\
\midrule
Inclusion indicator & 2.007** & 0.051 & 1.768* & 0.089 \\
 & (0.900) & & (0.971) & \\
Turnover slope: exclusions & -3.543*** & 0.010 & -3.365** & 0.083 \\
 & (1.089) & & (1.532) & \\
Turnover slope: inclusions & 0.086 & 0.925 & -0.027 & 0.979 \\
 & (0.884) & & (0.924) & \\
\bottomrule
\end{tabular}
\begin{tablenotes}[flushleft]
\footnotesize
\item The model parameterizes separate standardized log-turnover slopes for inclusions and exclusions, so the inclusion slope is reported directly rather than inferred from an interaction alone. Controls are pre-event size, return, volatility, and Amihud illiquidity. Standard errors and stars use SC-list-date clustering; $p_{wild}$ uses 999 restricted Rademacher replications.
\end{tablenotes}
\end{threeparttable}
\end{table}

%% file: table_routeb_ownership_did.tex
\begin{table}[!htbp]
\centering
\caption{Matched-control ownership response after continuously listed inclusions}
\label{tab:routeb_ownership_did}
\scriptsize
\begin{threeparttable}
\begin{adjustbox}{max width=\textwidth}
\begin{tabular}{lrrrr}
\toprule
Investor bucket & Treated change & Control change & Difference & $p_{wild}$ \\
\midrule
Shariah-sensitive portfolio & 0.022 & 0.000 & 0.021 & 0.360 \\
Broad mandate-constrained portfolio & 0.017 & -0.036 & 0.053 & 0.282 \\
Global passive placebo & -0.002 & -0.001 & -0.001 & 0.832 \\
\bottomrule
\end{tabular}
\end{adjustbox}
\begin{tablenotes}[flushleft]
\footnotesize
\item Entries are percentage-point changes in visible ownership from the dated pre- to post-release LSEG snapshots. The sample contains continuously listed, liquidity-qualified inclusions. Each treated change is compared with the mean change among its three continuously compliant matched controls. The portfolio buckets exclude likely strategic or control holders. Inference clusters by SC list date and uses 999 wild-cluster replications.
\end{tablenotes}
\end{threeparttable}
\end{table}

%% file: table_routeb_placebo_clean.tex
\begin{table}[!htbp]
\centering
\caption{Mid-review timing placebo for the primary Malaysia matched sets}
\label{tab:malaysia_placebo}
\begin{threeparttable}
\begin{tabular}{lrrrrr}
\toprule
Pseudo-event window & Treated events & Mean & Cluster SE & $p_{date}$ & $p_{wild}$ \\
\midrule
$[0,1]$ & 295 & 0.222 & 0.628 & 0.724 & 0.736 \\
$[0,3]$ & 295 & 0.243 & 0.771 & 0.753 & 0.738 \\
$[0,5]$ & 295 & 0.988 & 0.915 & 0.280 & 0.305 \\
$[0,10]$ & 295 & 0.841 & 1.151 & 0.465 & 0.473 \\
$[0,20]$ & 295 & 0.895 & 1.181 & 0.449 & 0.466 \\
\bottomrule
\end{tabular}
\begin{tablenotes}[flushleft]
\footnotesize
\item Each focal date is shifted three calendar months before the public SC release, placing it between the preceding and current semi-annual reviews. The test retains the primary treated securities and their fixed three-control matched sets, then reconstructs market-adjusted returns around the pseudo-date. It is a timing falsification rather than another classification event. Inference clusters by the focal SC list date and uses 999 wild-cluster replications. The older six-month shift is retained in the Internet Appendix and relabelled as a preceding-review falsification.
\end{tablenotes}
\end{threeparttable}
\end{table}

%% file: appendix_content.tex
\begin{center}
{\Large\bfseries Internet Appendix to\\[0.35em]
The Price of Permission: Classification Uncertainty in Constrained Capital Markets}\par
\vspace{1.1em}
Abdulrahman Qadi\quad Akash Sharma\quad Francesca Medda\par
\vspace{0.45em}
{\small Institute of Finance and Technology, University College London\\
Gower Street, London WC1E 6BT, United Kingdom}
\end{center}

\vspace{1.25em}
This Internet Appendix reports diagnostic tables and figures that support the main manuscript. The main manuscript is self-contained; this appendix preserves additional sample-construction, measurement, matching, placebo, and robustness evidence.

\clearpage
\appendix
\numberwithin{table}{section}
\numberwithin{figure}{section}
\clearpage
\section{Appendix Roadmap and Evidence Architecture}\label{app:roadmap}

This appendix preserves the supporting evidence required to evaluate sample construction, researcher-emulated rulebook measurement, matching quality, event-study validity, placebo behavior, and investor-base diagnostics. Route A documents the U.S. cross-standard design. Route B documents the SC Malaysia official-list design. The appendix distinguishes screen-implied transitions from official classifications and mechanism-consistent diagnostics from marginal-buyer identification.

\begin{table}[!htbp]
\centering
\caption{Mapping from theory to empirical evidence}
\label{tab:appendix_theory_map}
\begin{threeparttable}
\footnotesize
\begin{adjustbox}{max width=\textwidth}
\begin{tabular}{L{0.24\textwidth}L{0.34\textwidth}L{0.34\textwidth}}
\toprule
Theoretical object & Route A evidence & Route B evidence \\
\midrule
Eligibility vector $e_i$ & Seven researcher-emulated rulebooks produce security-month eligibility vectors. & Official SC Malaysia stock-code membership states across semi-annual lists. \\
Permitted investor mass $m_i$ & Not observed; equal-weighted eligibility share is reported only as rulebook coverage. & Official list status changes a recognized permission label but does not measure mandate capital. \\
Classification uncertainty & Cross-standard disagreement, nearest-boundary distance, and proximity risk. & Inclusion, exclusion, and continuous-compliance states define official-list events. \\
Monitoring implication & Disagreement and proximity rank next-month screen-implied transitions. & Not tested as a prediction exercise in Route B. \\
Pricing implication & Fama--MacBeth and DJIM/S\&P diagnostics show no unconditional rulebook-coverage premium. & Continuously listed, liquidity-qualified inclusions have positive matched CARs; mechanism diagnostics do not identify a marginal buyer. \\
\bottomrule
\end{tabular}
\end{adjustbox}
\begin{tablenotes}
\footnotesize
\item Route A and Route B use different evidence and estimands. Neither equal-weighted rulebook coverage nor official list status is a direct estimate of mandate-weighted investor mass.
\end{tablenotes}
\end{threeparttable}
\end{table}

\section{Route A Internet Appendix: U.S. Classification Uncertainty, Monitoring, and Methodology-Shock Diagnostics}\label{app:routeA_full}

This appendix documents the additional tables and figures produced by the empirical workflow. The main text reports the compact set of results needed for the paper's central claims. The additional outputs below are included to make clear which diagnostics support the construction of the sample, the classification-uncertainty measures, the matching procedure, the monitoring validation, and the daily event-study interpretation. Tables and figures that are purely diagnostic or too granular for the printed article are listed as replication-archive outputs rather than repeated in full.

\subsection{Data, measurement, and output inventory}\label{app:measurement}

Appendix Table \ref{tab:app_output_inventory} summarizes the saved empirical outputs. The archive contains both machine-readable CSV files and LaTeX versions for tables, and both PDF and PNG versions for figures where applicable.

\begin{longtable}{L{0.25\textwidth}L{0.68\textwidth}}
\caption{Replication output inventory}\label{tab:app_output_inventory}\\
\toprule
Category & Saved outputs \\
\midrule
\endfirsthead
\toprule
Category & Saved outputs \\
\midrule
\endhead
Sample construction and coverage & \texttt{table\_sample\_construction}, \texttt{table\_crsp\_coverage}. \\
Eligibility and margins & \texttt{table\_overall\_eligibility\_rates}, \texttt{table\_eligibility\_rates\_by\_year}, \texttt{table\_nearest\_standard\_counts}, \texttt{table\_standard\_margin\_summary}; figures for eligibility share, disagreement, proximity risk, nearest distance, and removed-screen margins. \\
Screen-implied transitions and monitoring & Two-way-clustered monitoring estimates, proximity-decile diagnostics, out-of-sample ranking performance, and standard-by-standard transition rates. \\
September 2023 groups and matching & \texttt{table\_sep2023\_event\_group\_counts}, \texttt{table\_sep2023\_newly\_eligible\_shock\_detail}, \texttt{table\_sep2023\_balance\_by\_group}, \texttt{table\_event\_group\_balance\_unmatched\_and\_matched}, \texttt{table\_event\_group\_balance\_newly\_both\_matched}, matching-pair and matching-weight tables for newly eligible and newly-both treatment groups. \\
Monthly event and panel diagnostics & \texttt{table\_sep2023\_monthly\_car\_regressions\_*}, \texttt{table\_sep2023\_monthly\_pretrend\_diagnostics}, \texttt{table\_sep2023\_did\_*}, \texttt{table\_dynamic\_event\_time\_*}, \texttt{table\_placebo\_sep\_window\_car\_tests}, \texttt{table\_panel\_market\_quality\_regressions}; monthly CAAR, placebo, and dynamic-event figures. \\
Daily event study & \texttt{table\_daily\_event\_main\_windows}, \texttt{table\_daily\_event\_pretrend\_windows}, \texttt{table\_daily\_dynamic\_joint\_pretrend\_tests}, \texttt{table\_daily\_dynamic\_event\_time\_coefficients}, \texttt{table\_daily\_event\_placebo\_focus}, \texttt{table\_daily\_event\_placebo\_all}, \texttt{table\_daily\_event\_car\_summary\_by\_group}, \texttt{table\_daily\_event\_diagnostic\_dashboard}; daily CAAR and dynamic treatment figures. \\
Expected-return diagnostics & \texttt{table\_fmb\_expected\_return\_diagnostics}, \texttt{table\_fmb\_expected\_return\_main\_variables}, \texttt{table\_fmb\_monthly\_coefficients}, \texttt{table\_fmb\_quality\_controls\_main}, \texttt{table\_fmb\_quality\_controls\_full}. \\
Additional robustness diagnostics & Mandate-weight scenarios, income-proxy sensitivity, predictive ranking, unified Route B sample flow, corrected threshold sensitivity, direct inclusion/exclusion turnover slopes, leave-one-date-out estimates, matched-control ownership changes, and the SC release-date audit. \\
September 2022 placebo diagnostics & \texttt{table\_sep2022\_placebo\_overall\_diagnostics}, \texttt{table\_sep2022\_placebo\_size\_diagnostics}, \texttt{table\_sep2022\_placebo\_sector\_diagnostics}, \texttt{table\_sep2022\_placebo\_market\_context}. \\
\bottomrule
\end{longtable}

\begin{table}[H]
\centering
\caption{Nearest-standard diagnostic counts}
\label{tab:app_nearest_standard}
\begin{tabular}{lr}
\toprule
Nearest standard & Firm-months \\
\midrule
FTSE/Yasaar & 378,435 \\
AAOIFI & 365,406 \\
SC Malaysia & 228,121 \\
DJIM & 126,183 \\
MSCI main & 94,444 \\
S\&P & 79,513 \\
MSCI M-Series & 70,504 \\
\bottomrule
\end{tabular}
\begin{minipage}{0.86\textwidth}
\footnotesize
\emph{Notes:} The table reports the standard associated with the minimum absolute active-boundary margin in the security-month panel. It is a diagnostic of which rulebook is locally closest to binding; it is not used as a standalone treatment variable.
\end{minipage}
\end{table}

\subsection{Rulebook implementation and measurement audits}\label{app:rulebook_audit}

The main manuscript reports the denominator, averaging window, and threshold used for every researcher-emulated rulebook. Business-activity exclusions are applied through a common SIC map. A required screen fails when its numerator or denominator is unavailable; nonpositive sales fail the income screen, and nonpositive total assets fail asset-denominated screens. The final Route A panel contains 113 nonpositive-total-asset observations and 41,339 nonpositive-sales observations. These rules prevent invalid ratios from being converted into artificial passes.

The income screen is the least exact Compustat emulation. The baseline numerator is positive interest income plus positive non-operating income. Because Compustat NOPI can contain permissible non-operating items, the baseline is deliberately conservative and overinclusive. Appendix Table~\ref{tab:impure_proxy_sensitivity} replaces it with positive interest income only. Agreement ranges from 94.46\% to 96.20\%, but average eligibility rises by 3.8--5.5 percentage points. The proxy therefore matters economically and should not be described as an official provider classification.

\input{table_impure_proxy_sensitivity}

Appendix Table~\ref{tab:routea_proxy_downstream} carries the alternative proxy through the monitoring and Fama--MacBeth constructions. The narrower proxy materially changes average eligibility and disagreement, but disagreement and proximity continue to rank next-month transitions. The fully adjusted eligibility slope remains economically small, while disagreement remains a residual characteristic association.

\input{table_routea_proxy_downstream_sensitivity}

\subsection{Additional measurement figures}\label{app:measurement_figures}

Appendix Figures \ref{fig:app_discrete_distributions} and \ref{fig:app_boundary_distributions} report diagnostic distributions for eligibility share, disagreement, boundary distance, proximity risk, and removed-screen margins.

\begin{figure}[p]
\centering
\begin{subfigure}{0.48\textwidth}
\centering
\includegraphics[width=\linewidth]{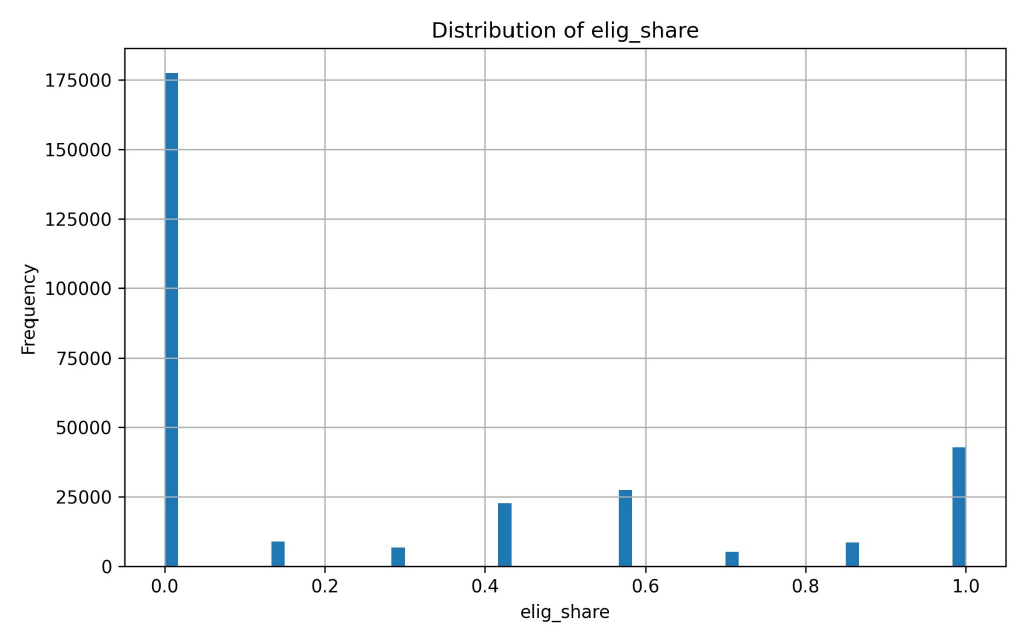}
\caption{Eligibility share}
\end{subfigure}
\begin{subfigure}{0.48\textwidth}
\centering
\includegraphics[width=\linewidth]{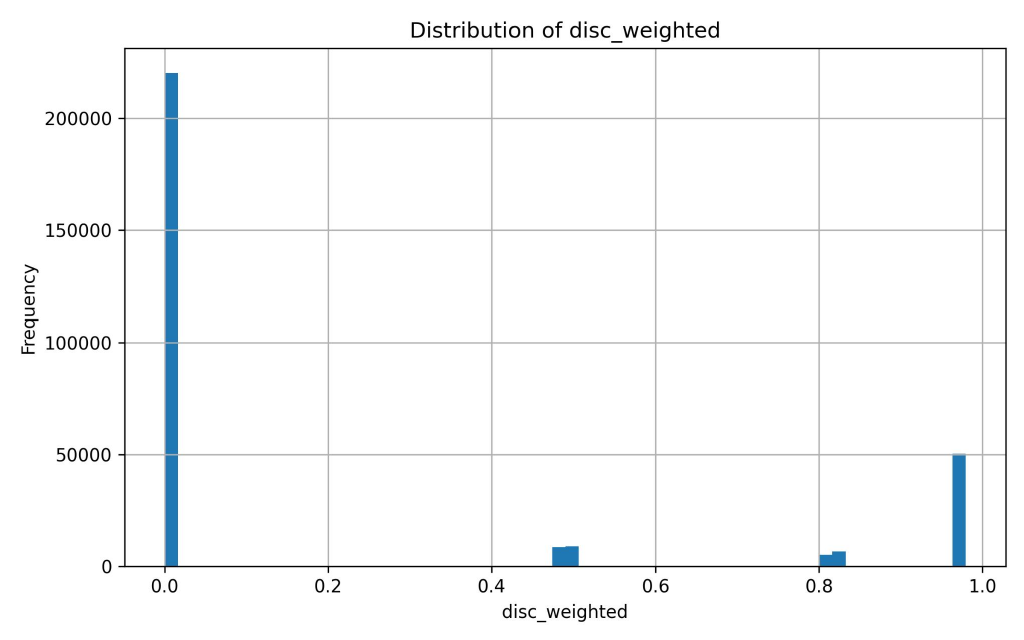}
\caption{Implemented disagreement}
\end{subfigure}
\caption{Distribution of equal-weighted classification variables}
\AltText{Two-panel histogram showing the security-month distributions of eligibility share and implemented disagreement.}
\label{fig:app_discrete_distributions}
\begin{minipage}{0.92\textwidth}
\footnotesize
\emph{Notes:} The variables are computed at the security-month level and are discrete under equal rulebook weights. They are rulebook-coverage diagnostics, not estimates of mandate-weighted investor mass.
\end{minipage}
\end{figure}

\begin{figure}[p]
\centering
\begin{subfigure}{0.48\textwidth}
\centering
\includegraphics[width=\linewidth]{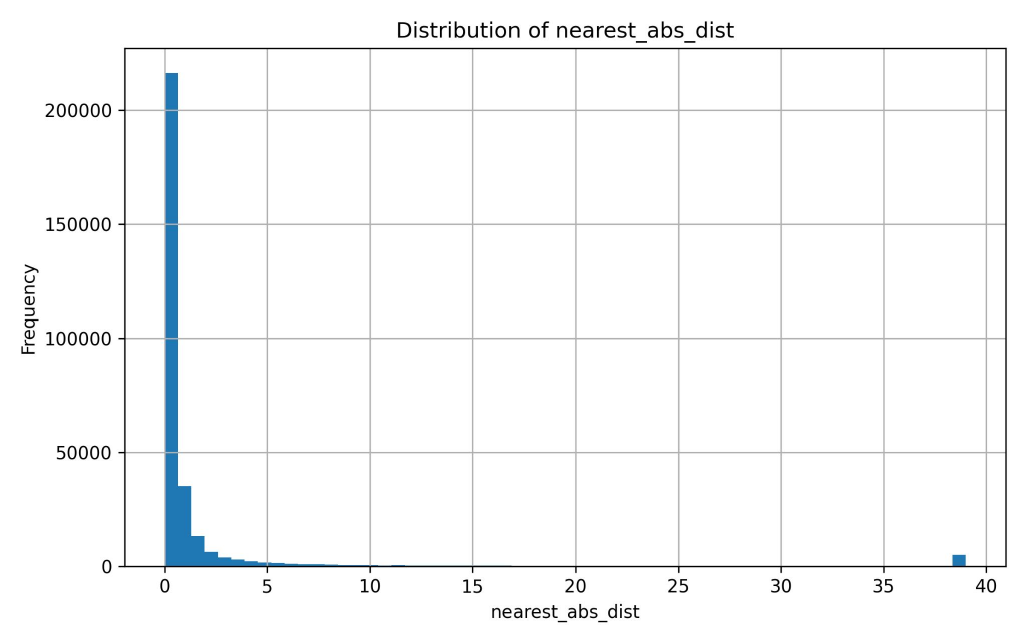}
\caption{Nearest absolute distance}
\end{subfigure}
\begin{subfigure}{0.48\textwidth}
\centering
\includegraphics[width=\linewidth]{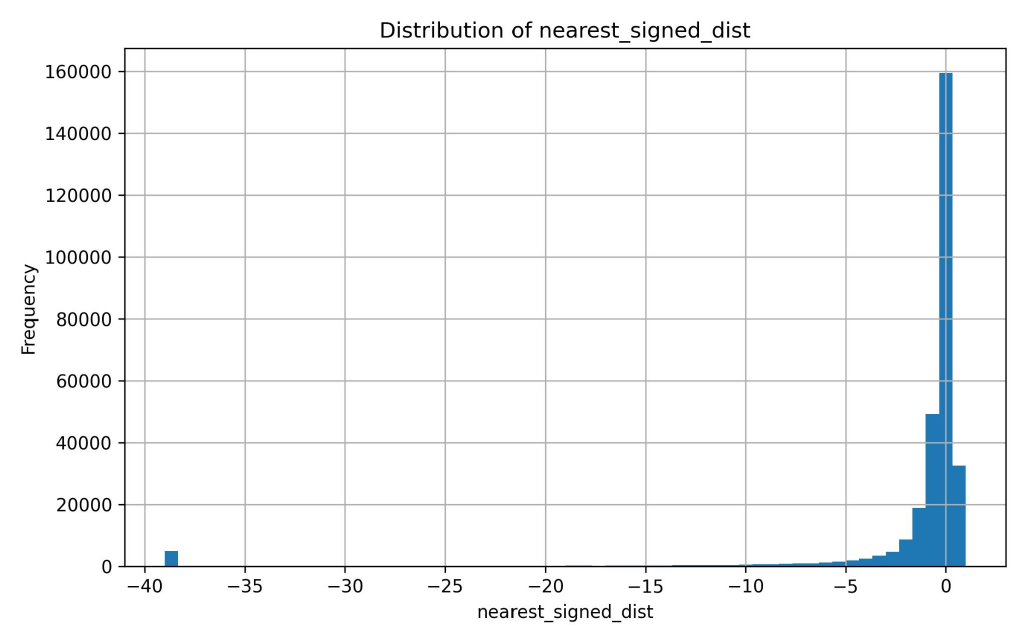}
\caption{Nearest signed distance}
\end{subfigure}
\begin{subfigure}{0.48\textwidth}
\centering
\includegraphics[width=\linewidth]{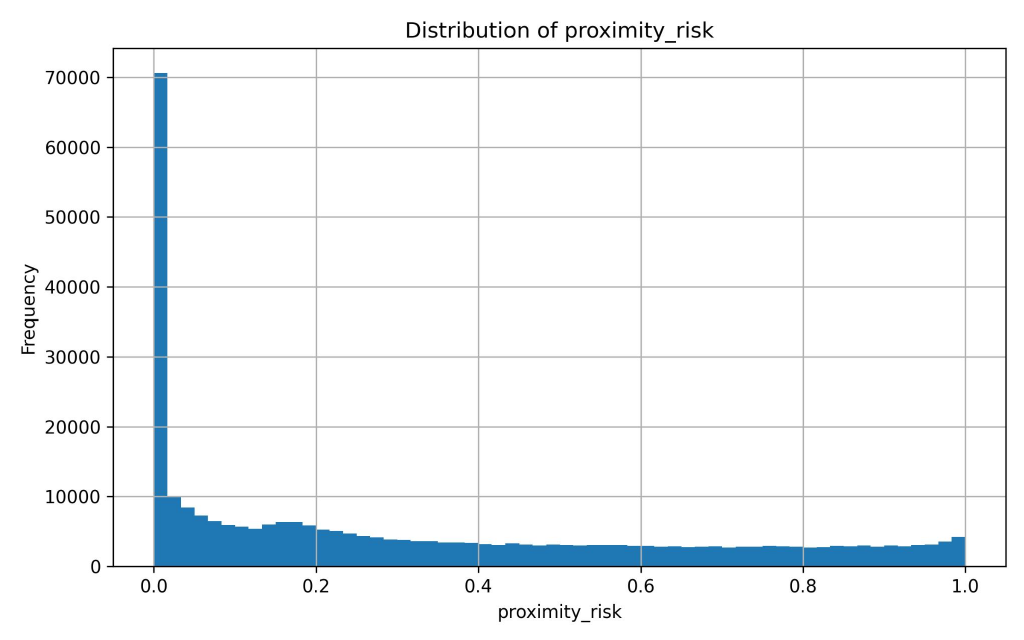}
\caption{Proximity risk}
\end{subfigure}
\begin{subfigure}{0.48\textwidth}
\centering
\includegraphics[width=\linewidth]{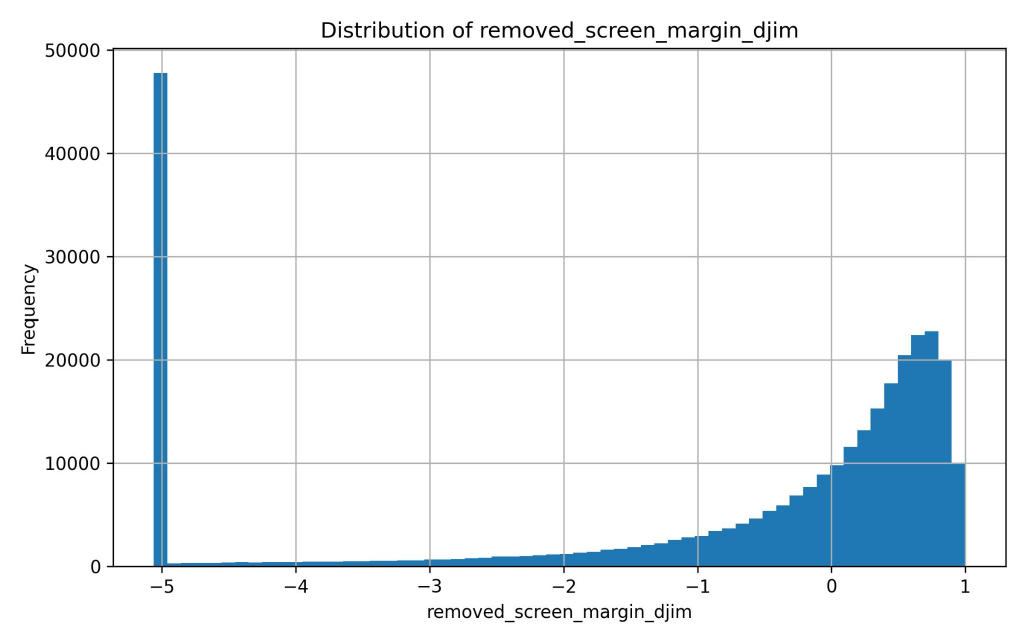}
\caption{DJIM removed-screen margin}
\end{subfigure}
\begin{subfigure}{0.48\textwidth}
\centering
\includegraphics[width=\linewidth]{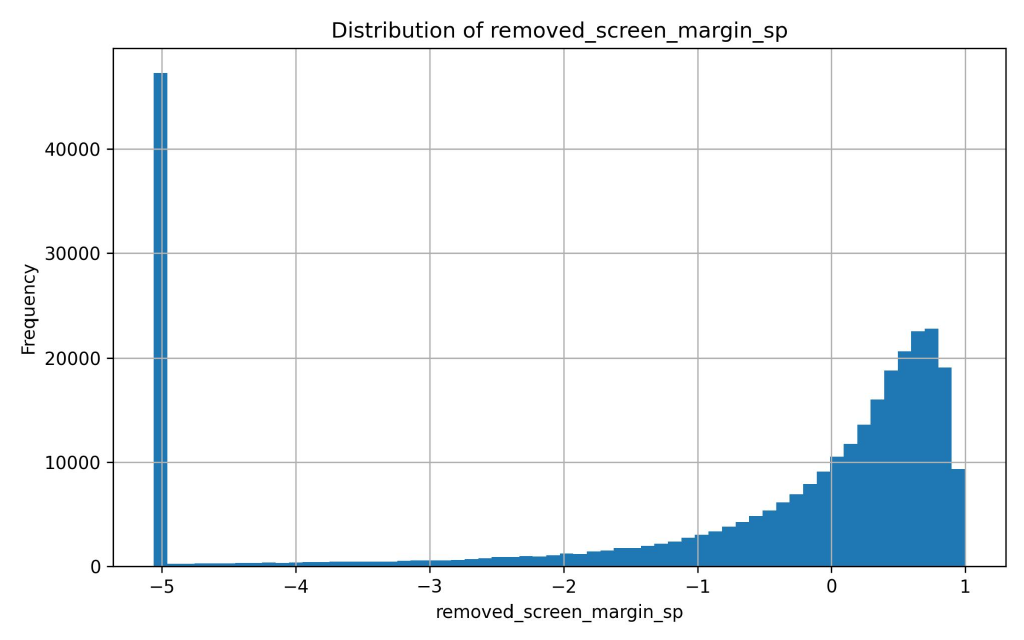}
\caption{S\&P removed-screen margin}
\end{subfigure}
\caption{Boundary-distance and removed-screen margin diagnostics}
\AltText{Five-panel histogram showing nearest absolute distance, nearest signed distance, proximity risk, and the DJIM and S\&P removed-screen margins used to construct boundary and methodology-shock measures.}
\label{fig:app_boundary_distributions}
\begin{minipage}{0.92\textwidth}
\footnotesize
\emph{Notes:} Removed-screen margins summarize distance from the cash and receivables screens removed in the 2023 DJIM/S\&P methodology change. These figures document the treatment construction and boundary-proximity variables used in the monitoring tests.
\end{minipage}
\end{figure}

\subsection{Monitoring validation}\label{app:monitoring}

The main text reports the compact transition logit. Additional decile tables and figures show that screen-implied transition frequency rises with proximity risk. These outputs are validation diagnostics rather than causal evidence because proximity and the outcome use the same implemented boundaries. The relevant saved outputs are:
\begin{itemize}
\item \path{table_reclassification_by_proximity_decile};
\item \path{fig_reclassification_by_proximity_decile};
\end{itemize}

Appendix Table~\ref{tab:app_standard_reclassification} decomposes the next-month screen-implied transition outcome by standard. AAOIFI has the highest security-month transition rate in the implemented panel, while the other rulebooks also contribute non-trivial changes.

\input{table_routea_standard_reclassification}

\subsection{Matching and balance diagnostics}\label{app:matching}

The event study uses matching because the raw September 2023 treatment groups are materially different. Continuously eligible firms are much larger and have lower cash and receivables ratios than newly eligible firms, while newly eligible firms are closer to the still-ineligible group. The replication archive reports unmatched and matched balance tables as well as matching-pair and matching-weight outputs for both newly eligible and newly-both treatment definitions. These outputs are listed in Appendix Table \ref{tab:app_output_inventory}. The daily event study in the main text uses the newly-both matched samples because the newly-both treatment is mechanically affected by both rule families.

\subsection{Additional daily event-study evidence}\label{app:daily}

\input{table_routea_daily_pretrend_corrected}

\subsection{Monthly and panel diagnostics}\label{app:monthly}

The replication workflow also produces monthly CAR regressions, monthly difference-in-differences specifications, placebo September-window tests, monthly dynamic event-time coefficients, and broad panel market-quality regressions. These results are intentionally not central to the paper. Monthly event windows are too coarse for the September 2023 methodology timing, and the panel market-quality regressions are sensitive because disagreement and proximity are mechanically related to the implemented labels. The main text therefore relies on the corrected daily matched event design and screen-implied transition validation.

\subsection{Fama--MacBeth expected-return diagnostics}\label{app:fmb}

Appendix Table \ref{tab:app_fmb_full} reports the expanded Fama--MacBeth diagnostics with standard asset-pricing controls. The main text adds stricter quality and screening-ratio controls. These are reduced-form characteristic regressions, not estimates of mandate-level capital weights.

\begin{table}[H]
\centering
\caption{Expanded Fama--MacBeth diagnostics}
\label{tab:app_fmb_full}
\begin{threeparttable}
\begin{adjustbox}{max width=\textwidth}
\begin{tabular}{llcc}
\toprule
Specification & Variable & Mean coefficient & $t$-statistic \\
\midrule
Eligibility share, controls & Eligibility share & 0.0013*** & 4.12 \\
 & Log market capitalization & 0.0013** & 2.25 \\
 & Log book-to-market & 0.0035*** & 3.37 \\
 & Momentum, 12-to-2 & 0.0045 & 1.42 \\
 & One-month reversal & -0.0166*** & -3.69 \\
\midrule
Disagreement, controls & Implemented disagreement & 0.0014*** & 4.47 \\
 & Log market capitalization & 0.0013** & 2.37 \\
 & Log book-to-market & 0.0034*** & 3.35 \\
 & Momentum, 12-to-2 & 0.0046 & 1.47 \\
 & One-month reversal & -0.0164*** & -3.64 \\
\bottomrule
\end{tabular}
\end{adjustbox}
\begin{tablenotes}[flushleft]
\footnotesize
\item The table reports time-series average coefficients from monthly Fama--MacBeth regressions of next-month excess returns. Variables are winsorized and the main permission variables are standardized. Standard errors are Newey--West with 12 lags. ***, **, and * denote significance at 1\%, 5\%, and 10\%.
\end{tablenotes}
\end{threeparttable}
\end{table}

\subsection{September 2022 placebo diagnostics}\label{app:placebo2022}

The strongest placebo result occurs around September 1, 2022. Appendix Tables \ref{tab:app_sep2022_overall}--\ref{tab:app_sep2022_market} examine this placebo window. The diagnostics show that the underperformance of the treated group is broad across size quartiles and not reducible to a single industry, although unclassified firms and business services account for a substantial share of the treated sample. The market-context table documents a severe selloff around the placebo window.

\begin{table}[H]
\centering
\caption{September 2022 placebo: overall diagnostics}
\label{tab:app_sep2022_overall}
\begin{threeparttable}
\begin{adjustbox}{max width=\textwidth}
\footnotesize
\begin{tabular}{lrrrr}
\toprule
Role & Firms & Weighted CAR[0,5] & Median CAR[0,5] & Mean log size \\
\midrule
Treated placebo & 233 & -2.95\% & -2.44\% & 5.59 \\
Still-ineligible controls & 403 & -0.43\% & -0.98\% & 5.75 \\
\bottomrule
\end{tabular}
\end{adjustbox}
\begin{tablenotes}[flushleft]
\footnotesize
\item CARs are market-model abnormal returns over the placebo September 2022 $[0,5]$ window. The treated-control difference is approximately -2.5 percentage points.
\end{tablenotes}
\end{threeparttable}
\end{table}

\begin{table}[H]
\centering
\caption{September 2022 placebo by size quartile}
\label{tab:app_sep2022_size}
\begin{threeparttable}
\begin{adjustbox}{max width=\textwidth}
\begin{tabular}{lrrr}
\toprule
Size quartile & Treated CAR[0,5] & Control CAR[0,5] & Difference \\
\midrule
Q1 small & -4.40\% & -2.22\% & -2.17\% \\
Q2 & -3.71\% & 0.73\% & -4.43\% \\
Q3 & -1.75\% & -0.14\% & -1.61\% \\
Q4 large & -1.71\% & -0.53\% & -1.17\% \\
\bottomrule
\end{tabular}
\end{adjustbox}
\begin{tablenotes}[flushleft]
\footnotesize
\item Size quartiles are formed using pre-event market capitalization. The placebo gap is larger in the lower half of the size distribution but remains negative in all four quartiles.
\end{tablenotes}
\end{threeparttable}
\end{table}

\begin{table}[H]
\centering
\caption{September 2022 placebo: largest treated sectors}
\label{tab:app_sep2022_sector}
\begin{threeparttable}
\begin{adjustbox}{max width=\textwidth}
\begin{tabular}{lrrrr}
\toprule
SIC2 & Treated firms & Treated share & Treated CAR[0,5] & Treated-control difference \\
\midrule
99 & 87 & 37.34\% & -0.84\% & -2.24\% \\
73 & 37 & 15.88\% & -4.66\% & -3.04\% \\
36 & 19 & 8.15\% & -5.25\% & -1.65\% \\
28 & 12 & 5.15\% & -5.68\% & -9.93\% \\
38 & 11 & 4.72\% & -2.93\% & 1.68\% \\
35 & 8 & 3.43\% & -3.71\% & -3.93\% \\
13 & 7 & 3.00\% & -5.29\% & -1.62\% \\
50 & 5 & 2.15\% & -5.61\% & -3.47\% \\
48 & 5 & 2.15\% & 0.16\% & 3.85\% \\
80 & 4 & 1.72\% & -3.40\% & -7.04\% \\
\bottomrule
\end{tabular}
\end{adjustbox}
\begin{tablenotes}[flushleft]
\footnotesize
\item The table lists the largest SIC2 categories in the treated placebo sample. The negative placebo effect is not driven by a single sector, although SIC2 99 and 73 account for a large share of treated firms.
\end{tablenotes}
\end{threeparttable}
\end{table}

\begin{table}[H]
\centering
\caption{September 2022 placebo: market context}
\label{tab:app_sep2022_market}
\begin{threeparttable}
\begin{adjustbox}{max width=\textwidth}
\begin{tabular}{lrr}
\toprule
Trading date & Event day & Cumulative market return \\
\midrule
August 31, 2022 & -1 & -4.35\% \\
September 1, 2022 & 0 & -4.46\% \\
September 9, 2022 & 5 & -1.70\% \\
September 13, 2022 & 7 & -4.67\% \\
September 23, 2022 & 15 & -10.95\% \\
\bottomrule
\end{tabular}
\end{adjustbox}
\begin{tablenotes}[flushleft]
\footnotesize
\item Cumulative market returns are computed from the CRSP daily market return series over the placebo event window. The placebo occurs during a volatile monetary-tightening episode following the late-August 2022 Jackson Hole repricing and the subsequent September selloff.
\end{tablenotes}
\end{threeparttable}
\end{table}

\section{Route B Internet Appendix: SC Malaysia Official-List Event Diagnostics}\label{app:routeB_full}

Appendix Table~\ref{tab:app_release_date_audit} reports the release-date override table used for the release-date-corrected event-study timing. The semi-annual list date remains the membership-period identifier and clustering unit; the public release date printed on the official PDF defines day 0.

\input{table_malaysia_release_date_audit}

\input{table_routeb_sample_flow_appendix}

\subsection{Matching and dependence diagnostics}\label{app:routeb_matching_details}

Appendix Table~\ref{tab:routeb_matching_audit} records the exact matching algorithm used by the release-date-corrected pipeline. It distinguishes the primary with-replacement nearest-neighbor design from the strict same-industry robustness, which drops treated events without three eligible same-SIC controls.

\input{table_routeb_matching_method_audit}

\input{table_routeb_sic2_exact}

\input{table_routeb_repeated_security}

The main manuscript reports the complete matched event-window profile and dynamic path. Cumulative pre-event CARs over $[-20,-1]$, $[-10,-1]$, and $[-5,-1]$ are individually insignificant. The joint list-date-clustered Wald test across all 20 daily pre-event coefficients rejects, so the appendix and main text do not treat the matched design as satisfying an unqualified parallel-trends condition.

\subsection{Preceding-review falsification}\label{app:preceding_review}

The older six-month shift is approximately the preceding official SC review, not a non-event date. Appendix Table~\ref{tab:app_preceding_review} is therefore retained only as a preceding-review falsification. Counts are treated-event counts after averaging each treated event's matched controls; they are not pair counts.

\begin{table}[!htbp]
\centering
\caption{Preceding-review falsification estimates}
\label{tab:app_preceding_review}
\begin{threeparttable}
\begin{tabular}{llrrrr}
\toprule
Treatment & Window & Treated events & Mean & $p_{date}$ & $p_{wild}$ \\
\midrule
Exclusion & $[0,1]$ & 368 & 0.80\% & 0.191 & 0.173 \\
Exclusion & $[0,5]$ & 368 & 0.64\% & 0.310 & 0.306 \\
Exclusion & $[0,10]$ & 368 & 0.90\% & 0.264 & 0.274 \\
Inclusion & $[0,1]$ & 399 & -0.25\% & 0.596 & 0.607 \\
Inclusion & $[0,5]$ & 399 & -0.07\% & 0.918 & 0.935 \\
Inclusion & $[0,10]$ & 399 & -0.32\% & 0.662 & 0.679 \\
\bottomrule
\end{tabular}
\begin{tablenotes}
\footnotesize
\item The pseudo-event is shifted approximately six months before the focal review and therefore overlaps the preceding official classification release. This table cannot establish behavior on an ordinary non-event date. The main text instead uses a three-month mid-review timing placebo on the corrected primary matched sets.
\end{tablenotes}
\end{threeparttable}
\end{table}

\begin{table}[!htbp]
\centering
\caption{Matching attrition diagnostics}
\label{tab:mal_attrition}
\begin{threeparttable}
\footnotesize
\begin{adjustbox}{max width=\textwidth}
\begin{tabular}{llrrrr}
\toprule
Event type & Variable & Matched $N$ & Unmatched $N$ & Matched mean & Unmatched mean \\
\midrule
Continuous & Stock-code length & 12,050 & 4,302 & 4.00 & 4.08 \\
Continuous & Name length & 12,050 & 4,302 & 11.65 & 13.45 \\
Exclusion & Stock-code length & 429 & 232 & 4.00 & 4.13 \\
Exclusion & Name length & 429 & 232 & 12.46 & 18.35 \\
Inclusion & Stock-code length & 598 & 261 & 4.00 & 4.23 \\
Inclusion & Name length & 598 & 261 & 12.60 & 19.40 \\
\bottomrule
\end{tabular}
\end{adjustbox}
\begin{tablenotes}[flushleft]
\footnotesize
\item Unmatched events have longer names, suggesting that matching attrition is related to name quality, listing changes, or more complex entity histories. The main estimates therefore apply to the reliably linked listed universe.
\end{tablenotes}
\end{threeparttable}
\end{table}

\begin{table}[!htbp]
\centering
\caption{Current-snapshot LSEG investor-base coverage by event type}
\label{tab:lseg_cov}
\scriptsize
\begin{adjustbox}{max width=\textwidth}
\begin{tabular}{lrrrrr}
\toprule
Event type & Events & With LSEG & Mean holders & State/pension & Shariah-sensitive \\
\midrule
Continuous compliant & 12,024 & 12,024 & 38.16 & 30.31\% & 30.36\% \\
Exclusion & 374 & 374 & 31.63 & 18.45\% & 19.25\% \\
Inclusion & 595 & 595 & 32.29 & 26.05\% & 27.73\% \\
\bottomrule
\end{tabular}
\end{adjustbox}
\begin{flushleft}
\footnotesize
The table uses current-snapshot LSEG holder classifications, not historical event-date ownership files. It is therefore a descriptive salience diagnostic only. State/pension and Shariah-sensitive columns report the fraction of events with at least one visible holder in the named category. The lower exclusion-event shares are consistent with the paper's interpretation that exclusion pressure depends on whether constrained investors were marginal holders before removal.
\end{flushleft}
\end{table}

\subsection{Historical LSEG holder reconstruction and demand-pressure diagnostics}\label{app:historical_lseg_mechanism}

The historical mechanism tests in the main manuscript are separate from the current-snapshot salience diagnostic in Table~\ref{tab:lseg_cov}. For the treated-event mechanism tests, we reconstruct dated LSEG holder snapshots for RIC-mapped SC Malaysia inclusion and exclusion events. The treated-event cache covers 945 inclusion or exclusion events and contains 55,180 holder rows. For the stricter predicted-demand test, we separately pull dated holder snapshots for the continuously compliant matched controls used in the release-date-corrected event study; this matched-control cache contains 178,058 holder rows.

For each event $i,d$ and investor bucket $b$, let $Own_{i,d,-}^{b}$ and $Own_{i,d,+}^{b}$ denote percentage ownership by bucket $b$ in the pre- and post-list snapshots, and let $NInv_{i,d,-}^{b}$ and $NInv_{i,d,+}^{b}$ denote the corresponding number of visible holders. The realized mechanism variables are
\[
\Delta Own_{i,d}^{b}=Own_{i,d,+}^{b}-Own_{i,d,-}^{b},
\qquad
\Delta NInv_{i,d}^{b}=NInv_{i,d,+}^{b}-NInv_{i,d,-}^{b}.
\]
For comparison with exclusions, the exploratory ownership-response specification is
\[
\Delta Y_{i,d}^{b}=\alpha+\beta\,Inclusion_{i,d}+X_{i,d}^{\prime}\theta+\varepsilon_{i,d},
\]
where $\Delta Y_{i,d}^{b}$ is either $\Delta Own_{i,d}^{b}$ or $\Delta NInv_{i,d}^{b}$, and $X_{i,d}$ contains pre-event turnover, market capitalization, pre-event return, volatility, and Amihud illiquidity. Standard errors are clustered by SC list date, and the reported small-cluster check is a restricted Rademacher wild-cluster bootstrap by list date.

The matched-control predicted-demand gap uses the actual continuously compliant controls matched to inclusion event $i,d$. Let $\mathcal{M}_{i,d}$ be that matched-control set. The ex ante permitted-ownership target and predicted ownership gap are
\[
TargetOwn_{i,d}^{b}
=
\frac{1}{|\mathcal{M}_{i,d}|}
\sum_{j\in\mathcal{M}_{i,d}} Own_{j,d,-}^{b},
\qquad
Gap_{i,d}^{b}
=
\max\{0,TargetOwn_{i,d}^{b}-Own_{i,d,-}^{b}\}.
\]
The price-link test regresses the treated-event matched abnormal return on $Gap_{i,d}^{b}$ and the same pre-event controls:
\[
\diffCAR_{i,d,[0,10]}=\alpha+\lambda\,Gap_{i,d}^{b}+X_{i,d}^{\prime}\theta+\varepsilon_{i,d}.
\]
The portfolio-only buckets exclude likely strategic or control holders before aggregating ownership. A cleaner matched-control diagnostic compares each treated ownership change with the mean change among the same three continuously compliant controls used in the price test. Appendix Table~\ref{tab:app_routeb_ownership_did} reports that comparison on 288 corrected primary events with dated coverage. Shariah-sensitive ownership rises among treated inclusions, but neither it nor broad constrained ownership differs significantly from matched controls under wild-cluster inference.

\input{table_routeb_ownership_did_appendix}

Because price pressure should depend on the amount of required buying relative to trading capacity, we also convert the matched-control ownership gap into an ADV-scaled demand-pressure variable:
\[
PredPressure_{i,d}^{b}
=
\log\left(
1+
\frac{Gap_{i,d}^{b}\times MCap_{i,d}^{pre}/100}
{ADV_{i,d}^{[-30,-1]}}
\right),
\]
where $MCap_{i,d}^{pre}$ is pre-event market capitalization and $ADV_{i,d}^{[-30,-1]}$ is average daily trading value over the pre-event window. We test whether this pressure variable predicts subsequent constrained ownership changes, matched abnormal returns, abnormal trading volume, and a pressure-by-turnover interaction. We define matched abnormal dollar volume over $[0,k]$ as the treated stock's log event-window dollar volume relative to its own pre-event expected volume, minus the average of the same measure for the three continuously compliant matched controls.

Appendix Table~\ref{tab:rof_demand_pressure} now uses the corrected continuously listed, turnover-qualified sample whenever the necessary LSEG variables are present. Predicted broad-constrained pressure is positively associated with subsequent constrained ownership change, but its CAR coefficient is negative and insignificant. The global-passive pressure coefficient is also negative, and realized buying pressure has a negative but insignificant CAR association. Post-event realized ownership is endogenous to event-window prices and cannot be interpreted causally. These signs are inconsistent with treating the available pressure proxies as a simple price-setting mechanism, even though the predicted proxy contains some first-stage information for ownership adjustment.

\input{table_rof_demand_pressure_mechanism}

\subsection{LSEG ETF capacity and aggregate flow diagnostics}\label{app:lseg_etf_capacity}

A narrow listed-ETF mechanism is testable with the LSEG fund data. For each SC release cycle, we pull ETF NAV and shares outstanding one business day before the release and at post-release business-day horizons. ETF size is measured as NAV times shares outstanding. This produces a direct capacity measure for Malaysia-focused Shariah ETFs even though historical security-level ETF weights are not available.

Appendix Table~\ref{tab:lseg_etf_capacity_scale} shows that the listed ETF channel is economically nontrivial relative to the trading capacity of some treated stocks. In the liquidity-qualified inclusion sample with ETF-share coverage, the core Malaysia Shariah ETFs have median pre-release assets of MYR 354.8 million. A deliberately conservative one-percent allocation from those ETFs equals 2.81 times the median treated stock's pre-event average daily trading value; current non-cash ETF holdings have a median weight of 5.56 percent.

\input{table_lseg_etf_capacity_scale}

\begin{figure}[!htbp]
\centering
\includegraphics[width=0.84\textwidth]{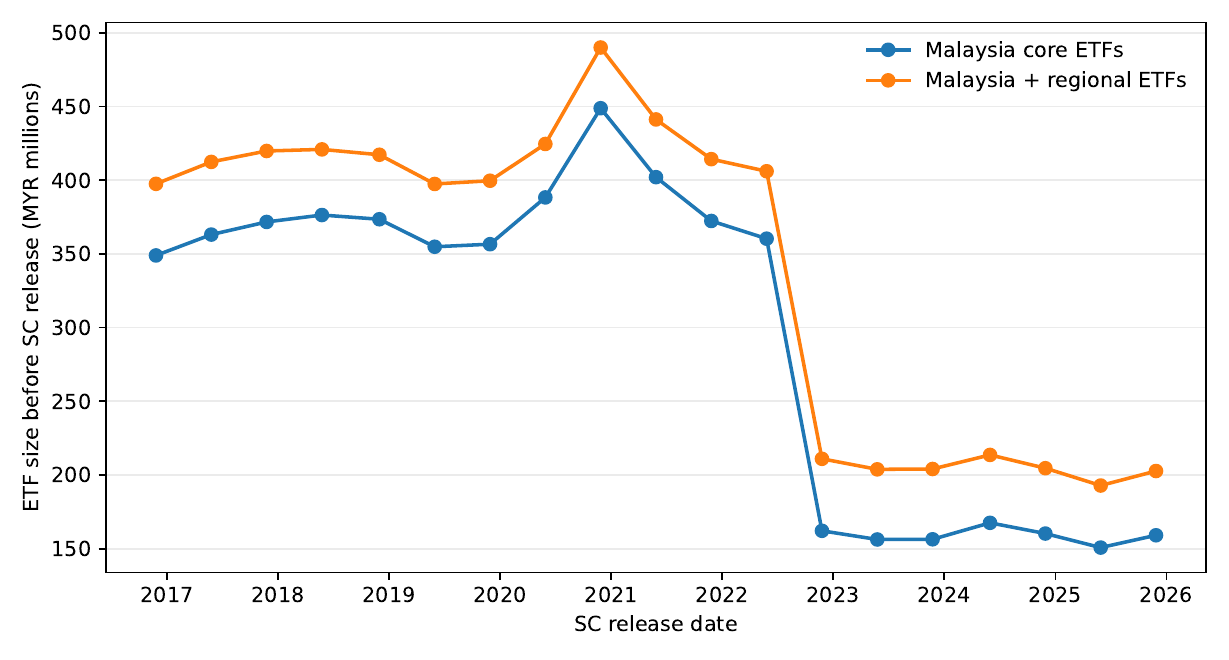}
\caption{LSEG Shariah ETF assets around SC Malaysia release cycles. ETF size is measured as NAV times shares outstanding one business day before each official SC release date.}
\AltText{Line chart showing Malaysia-focused and Malaysia-plus-regional Shariah ETF assets across SC Malaysia release cycles. Assets are observable from late 2016 onward and decline after 2022.}
\label{fig:lseg_etf_capacity}
\end{figure}

Appendix Table~\ref{tab:lseg_etf_capacity_tests} asks whether ETF capacity explains the cross-section of liquidity-qualified inclusion CARs. It does not. The coefficient on Malaysia core ETF capacity pressure is essentially zero over both $[0,10]$ and $[0,20]$, with wild-cluster $p$-values of 0.993 and 0.956. The Malaysia-plus-regional extension is also insignificant, and the capacity-by-turnover interaction is far from conventional significance. The Malaysia core ETF sample has no positive aggregate share-creation dates over the $[0,10]$ post-release window. These results make a narrow listed-ETF mechanism less likely; they do not identify another mechanism by elimination.

\input{table_lseg_etf_capacity_tests}

\subsection{Corrected fundamentals linkage and reclassification controls}\label{app:malaysia_fundamentals}

The original event--fundamentals merge read event-file GVKEYs such as \texttt{274603.0} while Compustat stored the same key as \texttt{274603}. Normalizing both sources to six-character strings restores 693 unique GVKEY overlaps. Of the matched event records, 834 have two fiscal observations before the official release: 507 inclusions and 327 exclusions. Appendix Table~\ref{tab:app_routeb_fundamentals} reports the resulting specification ladder. The common-sample regression uses treated-minus-matched-control changes in debt, cash, receivables, book leverage, log assets, and sales-to-assets. The cached panel does not contain usable ROA and the paper does not imply that it does. Because all controls are centered, the adjusted intercept equals the unadjusted 262-event common-sample mean; the larger coefficient relative to the 295-event row is a sample-composition result. The corrected merge therefore resolves the zero-observation diagnostic; the inconsistent former classification-surprise exercise is omitted.

\input{table_routeb_fundamentals_appendix}

\newpage

%% file: table_impure_proxy_sensitivity.tex
\begin{table}[!htbp]
\centering
\caption{Sensitivity to the non-permissible-income proxy}
\label{tab:impure_proxy_sensitivity}
\begin{threeparttable}
\begin{tabular}{lrrrr}
\toprule
Implementation & Baseline rate & Narrow-proxy rate & Agreement & Changed rows \\
\midrule
AAOIFI & 29.29\% & 34.81\% & 94.46\% & 74,360 \\
DJIM & 27.30\% & 32.74\% & 94.55\% & 73,232 \\
S\&P & 26.66\% & 31.84\% & 94.80\% & 69,813 \\
FTSE/Yasaar & 25.60\% & 29.40\% & 96.20\% & 51,026 \\
MSCI main & 28.37\% & 32.30\% & 96.06\% & 52,896 \\
MSCI M-Series & 26.69\% & 31.70\% & 94.97\% & 67,523 \\
SC Malaysia & 29.26\% & 33.16\% & 96.09\% & 52,465 \\
\bottomrule
\end{tabular}
\begin{tablenotes}[flushleft]
\footnotesize
\item The baseline proxy is $[\max(\mathrm{IDIT},0)+\max(\mathrm{NOPI},0)]/\mathrm{SALE}$. Because NOPI includes permissible non-operating items, the alternative uses $\max(\mathrm{IDIT},0)/\mathrm{SALE}$ only. Missing or nonpositive sales denominators fail the screen. Agreement reports identical binary labels.
\end{tablenotes}
\end{threeparttable}
\end{table}

%% file: table_routea_proxy_downstream_sensitivity.tex
\begin{table}[!htbp]
\centering
\caption{Downstream sensitivity to the positive-IDIT-only income proxy}
\label{tab:routea_proxy_downstream}
\scriptsize
\begin{threeparttable}
\begin{tabular}{llrr}
\toprule
Panel/test & Implementation or controls & Estimate & Test statistic \\
\midrule
Mean eligibility share & Baseline IDIT+NOPI & 0.2760 & SD=0.3788 \\
Mean eligibility share & Positive IDIT only & 0.3228 & SD=0.3827 \\
Mean disagreement & Baseline IDIT+NOPI & 0.2254 & SD=0.3879 \\
Mean disagreement & Positive IDIT only & 0.2886 & SD=0.4174 \\
\midrule
Monitoring: Implemented disagreement & Positive IDIT only & 1.0687*** & $z=41.27$ \\
Monitoring: Proximity risk & Positive IDIT only & 3.1006*** & $z=31.09$ \\
Monitoring: Log market capitalization & Positive IDIT only & -0.0682*** & $z=-10.36$ \\
\midrule
Fama--MacBeth: Eligibility share & Baseline & 0.0007* & $t=1.90$ \\
Fama--MacBeth: Eligibility share & + quality & 0.0002 & $t=0.48$ \\
Fama--MacBeth: Eligibility share & + quality + screens & -0.0005* & $t=-1.94$ \\
Fama--MacBeth: Implemented disagreement & Baseline & 0.0008** & $t=2.18$ \\
Fama--MacBeth: Implemented disagreement & + quality & 0.0007** & $t=2.11$ \\
Fama--MacBeth: Implemented disagreement & + quality + screens & 0.0006** & $t=2.42$ \\
\bottomrule
\end{tabular}
\begin{tablenotes}[flushleft]
\footnotesize
\item The alternative income screen uses $\max(\mathrm{IDIT},0)/\mathrm{SALE}$ and rebuilds all seven labels, equal-weighted eligibility and disagreement, active-boundary proximity, next-month transitions, and the Fama--MacBeth characteristics. Monitoring standard errors are two-way clustered by security and month. Fama--MacBeth standard errors are Newey--West with 12 lags. The fully adjusted alternative uses the IDIT-only income ratio as its screen-ratio control.
\end{tablenotes}
\end{threeparttable}
\end{table}

%% file: table_routea_standard_reclassification.tex
\begin{table}[!htbp]
\centering
\caption{Standard-by-standard next-month reclassification rates}
\label{tab:app_standard_reclassification}
\begin{tabular}{lcc}
\toprule
Standard & Reclassification events & Firm-month rate \\
\midrule
        AAOIFI & 35,063 & 2.61\% \\
        DJIM & 17,579 & 1.31\% \\
        S\&P & 16,296 & 1.21\% \\
        FTSE/Yasaar & 11,858 & 0.88\% \\
        MSCI main & 12,375 & 0.92\% \\
        MSCI M-Series & 16,084 & 1.20\% \\
        SC Malaysia & 12,576 & 0.94\% \\
\bottomrule
\end{tabular}
\begin{minipage}{0.92\textwidth}
\footnotesize \emph{Notes:} A reclassification event is a change in the binary eligibility indicator for the same CRSP security between adjacent firm-month observations. The table decomposes the all-standard monitoring outcome used in the main text.
\end{minipage}
\end{table}

%% file: table_routea_daily_pretrend_corrected.tex
\begin{table}[!htbp]
\centering
\caption{Route A daily pre-event cumulative abnormal returns}
\label{tab:app_daily_pretrend}
\begin{adjustbox}{max width=\textwidth}
\begin{tabular}{llrrr}
\toprule
Event date & Specification & CAR$[-20,-1]$ & CAR$[-10,-1]$ & CAR$[-5,-1]$ \\
\midrule
Announcement Aug. 4, 2023 & Unmatched newly-both & 0.441 & 0.694 & -0.805* \\
 & & (0.47) & (1.15) & (-1.76) \\
 & Matched vs. still-ineligible & -1.036 & -0.648 & -1.168* \\
 & & (-0.78) & (-0.77) & (-1.68) \\
 & Matched vs. continuous & 0.151 & 0.690 & -0.201 \\
 & & (0.08) & (0.48) & (-0.26) \\
\addlinespace
Pro-forma Sept. 1, 2023 & Unmatched newly-both & 1.583* & 0.022 & 0.398 \\
 & & (1.72) & (0.04) & (1.00) \\
 & Matched vs. still-ineligible & 0.699 & -0.846 & 0.342 \\
 & & (0.56) & (-1.06) & (0.59) \\
 & Matched vs. continuous & 1.371 & -0.237 & -0.135 \\
 & & (0.97) & (-0.27) & (-0.22) \\
\bottomrule
\end{tabular}
\end{adjustbox}
\begin{flushleft}
\footnotesize CARs are sums of daily market-model abnormal returns, winsorized at the 1st and 99th percentiles within each event-window cross-section. Reported effects are percentage-point coefficients with HC3 $t$-statistics.
\end{flushleft}
\end{table}

%% file: table_malaysia_release_date_audit.tex
\begin{longtable}{lll}
\caption{SC Malaysia list dates and release-date overrides}\label{tab:app_release_date_audit}\\
\toprule
List date & Public release date & Official PDF file \\
\midrule
\endfirsthead
\toprule
List date & Public release date & Official PDF file \\
\midrule
\endhead
        2013-11-30 & 2013-11-29 & \texttt{nov2013.pdf} \\
        2014-05-31 & 2014-05-30 & \texttt{may2014.pdf} \\
        2014-11-30 & 2014-11-28 & \texttt{nov2014.pdf} \\
        2015-05-31 & 2015-05-29 & \texttt{may2015.pdf} \\
        2015-11-30 & 2015-11-27 & \texttt{nov2015.pdf} \\
        2016-05-31 & 2016-05-27 & \texttt{may2016.pdf} \\
        2016-11-30 & 2016-11-25 & \texttt{nov2016.pdf} \\
        2017-05-31 & 2017-05-26 & \texttt{may2017.pdf} \\
        2017-11-30 & 2017-11-24 & \texttt{nov2017.pdf} \\
        2018-05-31 & 2018-05-25 & \texttt{may2018.pdf} \\
        2018-11-30 & 2018-11-30 & \texttt{nov2018.pdf} \\
        2019-05-31 & 2019-05-31 & \texttt{may2019.pdf} \\
        2019-11-30 & 2019-11-29 & \texttt{nov2019.pdf} \\
        2020-05-31 & 2020-05-29 & \texttt{may2020.pdf} \\
        2020-11-30 & 2020-11-27 & \texttt{nov2020.pdf} \\
        2021-05-31 & 2021-05-28 & \texttt{may2021.pdf} \\
        2021-11-30 & 2021-11-26 & \texttt{nov2021.pdf} \\
        2022-05-31 & 2022-05-27 & \texttt{may2022.pdf} \\
        2022-11-30 & 2022-11-25 & \texttt{nov2022.pdf} \\
        2023-05-31 & 2023-05-26 & \texttt{may2023.pdf} \\
        2023-11-30 & 2023-11-24 & \texttt{nov2023.pdf} \\
        2024-05-31 & 2024-05-31 & \texttt{may2024.pdf} \\
        2024-11-30 & 2024-11-29 & \texttt{nov2024.pdf} \\
        2025-05-31 & 2025-05-30 & \texttt{may2025.pdf} \\
        2025-11-30 & 2025-11-28 & \texttt{nov2025.pdf} \\
\bottomrule
\multicolumn{3}{p{0.92\textwidth}}{\footnotesize \emph{Notes:} Release dates are the public dates printed on the official Securities Commission Malaysia PDF cover pages and manually verified from first-page renders. The list date remains the semi-annual membership-period identifier and the clustering unit; the release date defines event day 0.}\\
\end{longtable}

%% file: table_routeb_sample_flow_appendix.tex
\begin{table}[!htbp]
\centering
\caption{Route B event-sample reconciliation}
\label{tab:app_routeb_sample_flow}
\begin{threeparttable}
\begin{tabular}{lrr}
\toprule
Panel A: cumulative event-sample stage & Inclusions & Exclusions \\
\midrule
Official stock-code transitions & 859 & 661 \\
Linked to a Compustat Global security & 598 & 429 \\
Valid release-date CAR & 595 & 375 \\
Successful same-date three-control match & 582 & 367 \\
Continuously listed before the preceding review & 410 & 367 \\
Continuously listed and above turnover floor & 295 & 257 \\
\midrule
Panel B: non-cumulative fundamentals coverage & Inclusions & Exclusions \\
\midrule
Matched events with two prior fiscal observations & 507 & 327 \\
Primary inclusions with complete relative controls & 262 & -- \\
\bottomrule
\end{tabular}
\begin{tablenotes}[flushleft]
\footnotesize
\item Panel A follows cumulative event selection. Panel B is deliberately separate because two-fiscal-observation coverage is measured among all matched linked events and is not a stage after the continuously listed restriction. Events are created from consecutive official stock-code lists before security matching. The valid-CAR row reports the 595 inclusions and 375 exclusions used in the raw event-return table.
\end{tablenotes}
\end{threeparttable}
\end{table}

%% file: table_routeb_matching_method_audit.tex
\begin{table}[!htbp]
\centering
\caption{Route B matching implementation audit}
\label{tab:routeb_matching_audit}
\begin{tabularx}{\textwidth}{L{0.31\textwidth}X}
\toprule
Design element & Implementation \\
\midrule
Candidate controls & Continuously compliant securities from the same SC list date \\
Covariates and windows & All use trading days $[-30,-1]$: log mean market capitalization; cumulative market-adjusted return; standard deviation of market-adjusted returns; mean daily turnover; and mean daily Amihud illiquidity \\
Scaling & Covariates standardized within each treated security plus its same-date candidate pool \\
Distance & Euclidean distance across the five standardized covariates \\
Selection & Three nearest controls, selected independently for each treated event \\
Replacement & With replacement; controls can serve multiple treated events and review cycles \\
Caliper & None \\
Missing values & Treated observations or candidates missing any matching covariate are excluded \\
Ties & Sorted by distance and then event identifier in the audited robustness implementation \\
Optimization & Greedy treated-by-treated nearest-neighbor selection, not global optimal assignment \\
\bottomrule
\end{tabularx}
\begin{flushleft}
\footnotesize The primary sample contains 295 treated events and 885 treated-control pairs. The strict SIC-2 robustness requires three same-industry candidates and never falls back to cross-industry controls.
\end{flushleft}
\end{table}

%% file: table_routeb_sic2_exact.tex
\begin{table}[!htbp]
\centering
\caption{Strict same-industry matching robustness}
\label{tab:routeb_sic2_exact}
\begin{threeparttable}
\begin{tabular}{lrrrrr}
\toprule
Window & Treated events & Estimate & Cluster SE & $p_{date}$ & $p_{wild}$ \\
\midrule
$[0,10]$ & 266 & 0.869 & 0.772 & 0.260 & 0.273 \\
$[0,20]$ & 266 & 1.481 & 1.138 & 0.193 & 0.232 \\
\bottomrule
\end{tabular}
\begin{tablenotes}[flushleft]
\footnotesize
\item This robustness design requires all three controls to share the treated security's two-digit SIC industry. Treated events with fewer than three same-date, same-SIC continuously compliant candidates are omitted; there is no fallback to cross-industry controls. Within the eligible pool, matching uses the same standardized Euclidean distance and allows control reuse.
\end{tablenotes}
\end{threeparttable}
\end{table}

%% file: table_routeb_repeated_security.tex
\begin{table}[!htbp]
\centering
\caption{Repeated-security inference and first-event robustness}
\label{tab:routeb_repeated_security}
\begin{threeparttable}
\begin{tabular}{llrrrr}
\toprule
Sample & Window & $N$ & Estimate & $p_{wild}$ & $p_{date,security}$ \\
\midrule
Primary sample & $[0,10]$ & 295 & 1.759*** & 0.017 & 0.008 \\
Primary sample & $[0,20]$ & 295 & 2.252** & 0.035 & 0.015 \\
First inclusion per security & $[0,10]$ & 237 & 1.688*** & 0.018 & 0.008 \\
First inclusion per security & $[0,20]$ & 237 & 2.405** & 0.032 & 0.020 \\
\bottomrule
\end{tabular}
\begin{tablenotes}[flushleft]
\footnotesize
\item $p_{date,security}$ uses two-way clustering by SC list date and treated security. The first-event sample retains only the earliest primary inclusion for each security. Wild-cluster inference by list date remains the principal small-cluster check. Controls are matched with replacement and may be reused; the replication output separately audits control reuse.
\end{tablenotes}
\end{threeparttable}
\end{table}

%% file: table_routeb_ownership_did_appendix.tex
\begin{table}[!htbp]
\centering
\caption{Matched-control ownership response after continuously listed inclusions}
\label{tab:app_routeb_ownership_did}
\scriptsize
\begin{threeparttable}
\begin{adjustbox}{max width=\textwidth}
\begin{tabular}{lrrrr}
\toprule
Investor bucket & Treated change & Control change & Difference & $p_{wild}$ \\
\midrule
Shariah-sensitive portfolio & 0.022 & 0.000 & 0.021 & 0.360 \\
Broad mandate-constrained portfolio & 0.017 & -0.036 & 0.053 & 0.282 \\
Global passive placebo & -0.002 & -0.001 & -0.001 & 0.832 \\
\bottomrule
\end{tabular}
\end{adjustbox}
\begin{tablenotes}[flushleft]
\footnotesize
\item Entries are percentage-point changes in visible ownership from the dated pre- to post-release LSEG snapshots. The sample contains continuously listed, liquidity-qualified inclusions. Each treated change is compared with the mean change among its three continuously compliant matched controls. The portfolio buckets exclude likely strategic or control holders. Inference clusters by SC list date and uses 999 wild-cluster replications.
\end{tablenotes}
\end{threeparttable}
\end{table}

%% file: table_rof_demand_pressure_mechanism.tex
\begin{table}[!htbp]
\centering
\caption{Demand-pressure and trading-volume mechanism tests}
\label{tab:rof_demand_pressure}
\begin{adjustbox}{max width=\textwidth}
\begin{tabular}{lllrrrrr}
\toprule
Test & Bucket/outcome & Term & Coef. & SE & $p_{date}$ & $p_{wild}$ & N \\
\midrule
Inclusion-volume response & All events: Abnormal dollar volume $[0,10]$ & Inclusion & 0.029 & 0.106 & 0.785 & 0.802 & 546 \\
Predicted pressure & Broad constrained, portfolio only: $\Delta$ constrained ownership & Predicted pressure & 0.059 & 0.028 & 0.038 & 0.037 & 288 \\
Predicted pressure & Broad constrained, portfolio only: $CAR[0,10]$ & Predicted pressure & -0.006 & 0.005 & 0.226 & 0.264 & 288 \\
Predicted pressure & Broad constrained, portfolio only: $CAR[0,10]$ & Predicted pressure $\times$ turnover & 0.009 & 0.014 & 0.528 & 0.524 & 288 \\
Predicted pressure & Global passive placebo: $CAR[0,10]$ & Predicted pressure & -0.007 & 0.004 & 0.095 & 0.145 & 288 \\
Realized pressure & Broad constrained, portfolio only: $CAR[0,10]$ & Realized buying pressure & -0.004 & 0.004 & 0.309 & 0.300 & 288 \\
Realized pressure & Broad constrained, portfolio only: Abnormal dollar volume $[0,10]$ & Realized buying pressure & 0.095 & 0.051 & 0.062 & 0.090 & 285 \\
\bottomrule
\end{tabular}
\end{adjustbox}
\begin{minipage}{0.95\linewidth}
\footnotesize \emph{Notes:} The table reports demand-pressure mechanism diagnostics on the corrected continuously listed, turnover-qualified sample whenever the required LSEG variables are available. Predicted pressure is the matched-control constrained ownership gap converted to currency value using pre-event market capitalization and scaled by pre-event average daily trading value, then log-transformed and standardized within the estimating sample. Realized buying pressure uses the positive post-minus-pre LSEG constrained ownership change, converted to value and scaled by pre-event average daily trading value. All specifications use release-date event timing, pre-event controls, list-date clustered standard errors, and restricted Rademacher wild-cluster p-values by SC list date with 999 replications. Opposite-signed price coefficients are reported rather than treated as supportive mechanism evidence. Realized-pressure rows are endogenous mechanism associations because post-event holder snapshots are observed after the price window.
\end{minipage}
\end{table}

%% file: table_lseg_etf_capacity_scale.tex
\begin{table}[!htbp]
\centering
\caption{LSEG Shariah ETF channel capacity}
\label{tab:lseg_etf_capacity_scale}
\begin{adjustbox}{max width=\textwidth}
\begin{tabular}{lrrrrrrr}
\toprule
Channel & Events & Dates & Median size & Latest size & Median 1\% pressure & Median current weight & Nonzero creations \\
\midrule
Malaysia core ETFs & 378 & 19 & 354.826 & 159.099 & 2.807 & 5.560 & 0 \\
Malaysia + regional ETFs & 378 & 19 & 397.496 & 202.688 & 3.518 & 6.689 & 0 \\
\bottomrule
\end{tabular}
\end{adjustbox}
\begin{minipage}{0.95\linewidth}
\footnotesize \emph{Notes:} ETF size is measured from LSEG NAV times ETF shares outstanding one business day before the official SC Malaysia release date. Median and latest size are in MYR millions. Median 1\% pressure is the hypothetical purchase value from allocating one percent of channel ETF assets to a treated stock, scaled by that stock's pre-event average daily trading value. Median current weight is the median current constituent weight among non-cash ETF holdings observed in LSEG. This is a capacity stress test, not an observed security-level predicted purchase. Nonzero creations counts list dates with positive aggregate ETF share-creation value over the post-release $[0,10]$ business-day window.
\end{minipage}
\end{table}

%% file: table_lseg_etf_capacity_tests.tex
\begin{table}[!htbp]
\centering
\caption{LSEG ETF capacity and flow-channel tests}
\label{tab:lseg_etf_capacity_tests}
\begin{adjustbox}{max width=\textwidth}
\begin{tabular}{lllrrrrr}
\toprule
Channel & Outcome & Test variable & Coef. & SE & $p_{date}$ & $p_{wild}$ & N \\
\midrule
Malaysia core & $CAR[0,10]$ & 1\% capacity pressure & 0.000 & 0.013 & 0.990 & 0.993 & 378 \\
Malaysia core & $CAR[0,20]$ & 1\% capacity pressure & 0.001 & 0.012 & 0.957 & 0.956 & 378 \\
Malaysia + regional & $CAR[0,10]$ & 1\% capacity pressure & 0.000 & 0.013 & 0.985 & 0.981 & 378 \\
Malaysia core & $CAR[0,10]$ & Capacity pressure $\times$ turnover & 0.005 & 0.013 & 0.674 & 0.707 & 378 \\
\bottomrule
\end{tabular}
\end{adjustbox}
\begin{minipage}{0.95\linewidth}
\footnotesize \emph{Notes:} The table reports liquidity-qualified inclusion-event regressions using release-date event timing. ETF capacity pressure equals one percent of channel ETF assets measured one business day before the SC release, scaled by the treated stock's pre-event average daily trading value, log-transformed and standardized. The Malaysia ETF channels have no positive aggregate ETF share-creation dates over the post-release $[0,10]$ business-day window, so creation-flow regressions are not reported in the displayed rows. Specifications include pre-event market capitalization, return, volatility, and illiquidity controls. $p_{date}$ clusters by SC list date. $p_{wild}$ is a restricted Rademacher wild-cluster bootstrap by list date with 999 replications.
\end{minipage}
\end{table}

%% file: table_routeb_fundamentals_appendix.tex
\begin{table}[!htbp]
\centering
\caption{SC Malaysia inclusion specification ladder}
\label{tab:app_routeb_fundamentals}
\scriptsize
\begin{threeparttable}
\begin{adjustbox}{max width=\textwidth}
\begin{tabular}{lrrrrrrr}
\toprule
 & & \multicolumn{3}{c}{CAR$[0,10]$} & \multicolumn{3}{c}{CAR$[0,20]$} \\
\cmidrule(lr){3-5}\cmidrule(lr){6-8}
Sample/specification & $N$ & Estimate & $p_{date}$ & $p_{wild}$ & Estimate & $p_{date}$ & $p_{wild}$ \\
\midrule
All matched inclusions & 582 & 0.784 & 0.163 & 0.163 & 1.436* & 0.072 & 0.083 \\
Continuously listed only & 410 & 0.896 & 0.127 & 0.134 & 1.458* & 0.053 & 0.064 \\
Continuously listed + turnover floor & 295 & 1.759*** & 0.008 & 0.017 & 2.252** & 0.018 & 0.035 \\
Complete fundamentals sample, unadjusted & 262 & 2.154*** & 0.003 & 0.009 & 2.753*** & 0.005 & 0.012 \\
Complete sample + relative fundamentals & 262 & 2.154*** & 0.001 & 0.006 & 2.753*** & 0.003 & 0.017 \\
\bottomrule
\end{tabular}
\end{adjustbox}
\begin{tablenotes}[flushleft]
\footnotesize
\item Estimates are percentage-point means of treated-minus-three-control market-adjusted CARs. The continuously listed restriction requires trading before the preceding SC review. The turnover floor is fixed at 0.000296 from the pre-event matched universe. The complete sample requires all six treated and matched-control-relative fundamental changes and four market controls. The final row adds standardized treated-minus-control changes in debt, cash, receivables, book leverage, log assets, and sales-to-assets plus pre-event size, return, volatility, and Amihud illiquidity. Because controls are centered, its intercept is the common-sample mean; the comparison does not attribute the larger complete-sample mean to adding controls. Inference clusters by SC list date; $p_{wild}$ uses 999 restricted Rademacher replications.
\end{tablenotes}
\end{threeparttable}
\end{table}

%% file: references.bib
@techreport{AAOIFI2024,
  author      = {{Accounting and Auditing Organization for Islamic Financial Institutions}},
  title       = {{Shari'ah Standards}},
  institution = {AAOIFI},
  year        = {2024},
  address     = {Manama}
}

@article{AlnamlahHassan2022,
  author  = {Alnamlah, Abdullah and Hassan, M. Kabir and Alhomaidi, Asem and Smolo, Edib},
  title   = {A New Model for Screening {Shariah-Compliant} Firms},
  journal = {Borsa Istanbul Review},
  year    = {2022},
  volume  = {22},
  number  = {S1},
  pages   = {S10--S23},
  doi     = {10.1016/j.bir.2022.10.011}
}

@article{AshrafMohammad2016,
  author  = {Ashraf, Dawood and Mohammad, Nazeeruddin},
  title   = {Matching Perception with Reality: Performance of {Islamic} Equity Investments},
  journal = {Pacific-Basin Finance Journal},
  year    = {2016},
  volume  = {39},
  pages   = {175--189}
}

@article{AvramovCheng2022,
  author  = {Avramov, Doron and Cheng, Si and Lioui, Abraham and Tarelli, Andrea},
  title   = {Sustainable Investing with {ESG} Rating Uncertainty},
  journal = {Journal of Financial Economics},
  year    = {2022},
  volume  = {145},
  number  = {2},
  pages   = {642--664},
  doi     = {10.1016/j.jfineco.2021.09.009}
}

@article{AyedhEchchabi2019,
  author  = {Ayedh, Abdullah Mohammed Ahmed and Echchabi, Abdelghani and Aziz, Mohd Rizal Alif Abdul and Dandis, Mohammed Omar},
  title   = {{Shariah} Screening Methodology: Does It Really {Shariah} Compliance?},
  journal = {Iqtishadia: Jurnal Kajian Ekonomi dan Bisnis Islam},
  year    = {2019},
  volume  = {12},
  number  = {2},
  pages   = {144--172}
}

@article{BakerHollifieldOsambela2022,
  author  = {Baker, Steven D. and Hollifield, Burton and Osambela, Emilio},
  title   = {Asset Prices and Portfolios with Externalities},
  journal = {Review of Finance},
  year    = {2022},
  volume  = {26},
  number  = {6},
  pages   = {1433--1468},
  doi     = {10.1093/rof/rfac065}
}

@article{BergKolbel2022,
  author  = {Berg, Florian and Koelbel, Julian F. and Rigobon, Roberto},
  title   = {Aggregate Confusion: The Divergence of {ESG} Ratings},
  journal = {Review of Finance},
  year    = {2022},
  volume  = {26},
  number  = {6},
  pages   = {1315--1344},
  doi     = {10.1093/rof/rfac033}
}

@article{BerkvanBinsbergen2025,
  author  = {Berk, Jonathan B. and van Binsbergen, Jules H.},
  title   = {The Impact of Impact Investing},
  journal = {Journal of Financial Economics},
  year    = {2025},
  volume  = {164},
  pages   = {103972}
}

@article{CeccarelliRamelliWagner2024,
  author  = {Ceccarelli, Marco and Ramelli, Stefano and Wagner, Alexander F.},
  title   = {Low Carbon Mutual Funds},
  journal = {Review of Finance},
  year    = {2024},
  volume  = {28},
  number  = {1},
  pages   = {45--74},
  doi     = {10.1093/rof/rfad015}
}

@article{ChristensenSerafeim2022,
  author  = {Christensen, Dane M. and Serafeim, George and Sikochi, Anywhere},
  title   = {Why Is Corporate Virtue in the Eye of the Beholder? The Case of {ESG} Ratings},
  journal = {The Accounting Review},
  year    = {2022},
  volume  = {97},
  number  = {1},
  pages   = {147--175}
}

@techreport{DJIM2024,
  author      = {{S\&P Dow Jones Indices}},
  title       = {{Dow Jones Islamic Market Indices Methodology}},
  institution = {S\&P Dow Jones Indices},
  year        = {2024},
  address     = {New York}
}

@article{FabozziMa2008,
  author  = {Fabozzi, Frank J. and Ma, K. C. and Oliphant, Becky J.},
  title   = {Sin Stock Returns},
  journal = {The Journal of Portfolio Management},
  year    = {2008},
  volume  = {35},
  number  = {1},
  pages   = {82--94}
}

@article{FamaFrench1992,
  author  = {Fama, Eugene F. and French, Kenneth R.},
  title   = {The Cross-Section of Expected Stock Returns},
  journal = {The Journal of Finance},
  year    = {1992},
  volume  = {47},
  number  = {2},
  pages   = {427--465}
}

@article{FeldhutterPedersen2025,
  author  = {Feldhuetter, Peter and Pedersen, Lasse Heje},
  title   = {Is Capital Structure Irrelevant with {ESG} Investors?},
  journal = {The Review of Financial Studies},
  year    = {2025},
  volume  = {38},
  number  = {8},
  pages   = {2362--2385},
  doi     = {10.1093/rfs/hhae059}
}

@techreport{FTSE2022,
  author      = {{FTSE Russell}},
  title       = {{FTSE Shariah Global Equity Index Series Ground Rules}},
  institution = {FTSE Russell},
  year        = {2022},
  address     = {London}
}

@article{GibsonBrandon2021,
  author  = {Gibson Brandon, Rajna and Krueger, Philipp and Schmidt, Peter S.},
  title   = {{ESG} Rating Disagreement and Stock Returns},
  journal = {Financial Analysts Journal},
  year    = {2021},
  volume  = {77},
  number  = {4},
  pages   = {104--127}
}

@article{HartzmarkSussman2019,
  author  = {Hartzmark, Samuel M. and Sussman, Abigail B.},
  title   = {Do Investors Value Sustainability? A Natural Experiment Examining Ranking and Fund Flows},
  journal = {The Journal of Finance},
  year    = {2019},
  volume  = {74},
  number  = {6},
  pages   = {2789--2837},
  doi     = {10.1111/jofi.12841}
}

@article{HeinkelKraus2001,
  author  = {Heinkel, Robert and Kraus, Alan and Zechner, Josef},
  title   = {The Effect of Green Investment on Corporate Behavior},
  journal = {Journal of Financial and Quantitative Analysis},
  year    = {2001},
  volume  = {36},
  number  = {4},
  pages   = {431--449}
}

@article{HoRahmanYusufZamzamin2014,
  author  = {Ho, Catherine Soke Fun and Abd Rahman, Nurul Afiqah and Yusuf, Nor Azlan Mohamed and Zamzamin, Zulkifli},
  title   = {Performance of Global {Islamic} versus Conventional Share Indices: International Evidence},
  journal = {Pacific-Basin Finance Journal},
  year    = {2014},
  volume  = {28},
  pages   = {110--121}
}

@article{HongKacperczyk2009,
  author  = {Hong, Harrison and Kacperczyk, Marcin},
  title   = {The Price of Sin: The Effects of Social Norms on Markets},
  journal = {Journal of Financial Economics},
  year    = {2009},
  volume  = {93},
  number  = {1},
  pages   = {15--36}
}

@article{HouXue2015,
  author  = {Hou, Kewei and Xue, Chen and Zhang, Lu},
  title   = {Digesting Anomalies: An Investment Approach},
  journal = {The Review of Financial Studies},
  year    = {2015},
  volume  = {28},
  number  = {3},
  pages   = {650--705}
}

@article{KhatkhatayNisar2007,
  author  = {Khatkhatay, Mohammed Husain and Nisar, Shariq},
  title   = {{Shariah} Compliant Equity Investments: An Assessment of Current Screening Norms},
  journal = {Islamic Economic Studies},
  year    = {2007},
  volume  = {15},
  number  = {1},
  pages   = {47--76}
}

@article{Merton1987,
  author  = {Merton, Robert C.},
  title   = {A Simple Model of Capital Market Equilibrium with Incomplete Information},
  journal = {The Journal of Finance},
  year    = {1987},
  volume  = {42},
  number  = {3},
  pages   = {483--510}
}

@techreport{MSCI2024,
  author      = {{MSCI}},
  title       = {{MSCI Islamic Index Series Methodology}},
  institution = {MSCI Inc.},
  year        = {2024},
  address     = {New York}
}

@article{OrhanIsiker2021,
  author  = {Orhan, Zeyneb Hafsa and Isiker, Murat},
  title   = {Developing a Ranking Methodology for {Shari'ah} Indices: The Case of {Borsa Istanbul}},
  journal = {ISRA International Journal of Islamic Finance},
  year    = {2021},
  volume  = {13},
  number  = {3},
  pages   = {302--317}
}

@article{PastorStambaugh2021,
  author  = {Pastor, Lubos and Stambaugh, Robert F. and Taylor, Lucian A.},
  title   = {Sustainable Investing in Equilibrium},
  journal = {Journal of Financial Economics},
  year    = {2021},
  volume  = {142},
  number  = {2},
  pages   = {550--571}
}

@article{PedersenFitzgibbons2021,
  author  = {Pedersen, Lasse Heje and Fitzgibbons, Shaun and Pomorski, Lukasz},
  title   = {Responsible Investing: The {ESG-Efficient} Frontier},
  journal = {Journal of Financial Economics},
  year    = {2021},
  volume  = {142},
  number  = {2},
  pages   = {572--597}
}

@article{RizaldyAhmed2019,
  author  = {Rizaldy, M. Reyhan and Ahmed, Habib},
  title   = {{Islamic} Legal Methodologies and {Shariah} Screening Standards: Application in the {Indonesian} Stock Market},
  journal = {Thunderbird International Business Review},
  year    = {2019},
  volume  = {61},
  number  = {5},
  pages   = {733--749}
}

@techreport{SCMalaysia2013,
  author      = {{Securities Commission Malaysia}},
  title       = {Revised {Shariah} Screening Methodology for {Shariah-Compliant} Securities},
  institution = {Securities Commission Malaysia},
  year        = {2013},
  address     = {Kuala Lumpur}
}

@misc{SPDJIDJIMAnnouncement2023,
  author       = {{S\&P Dow Jones Indices}},
  title        = {{Dow Jones Islamic Market Indices Compliance Criteria Update \& September 2023 Rebalance Implementation}},
  year         = {2023},
  month        = {August},
  note         = {Index announcement dated August 4, 2023; accessed June 26, 2026},
  url          = {https://www.spglobal.com/spdji/en/documents/indexnews/announcements/20230804-1465480/1465480_djislamicmarketindices-8-4-2023.pdf}
}

@misc{SPDJIShariahAnnouncement2023,
  author       = {{S\&P Dow Jones Indices}},
  title        = {{S\&P Shariah Indices Compliance Criteria Update \& September 2023 Rebalance Implementation}},
  year         = {2023},
  month        = {August},
  note         = {Index announcement dated August 4, 2023; accessed June 26, 2026},
  url          = {https://www.spglobal.com/spdji/en/documents/indexnews/announcements/20230804-1465479/1465479_spshariahindices-8-4-2023.pdf}
}

@techreport{SP2024,
  author      = {{S\&P Dow Jones Indices}},
  title       = {{S\&P Shariah Indices Methodology}},
  institution = {S\&P Dow Jones Indices},
  year        = {2024},
  address     = {New York}
}

@article{Zerbib2022,
  author  = {Zerbib, Olivier David},
  title   = {A Sustainable Capital Asset Pricing Model ({S-CAPM}): Evidence from Environmental Integration and Sin Stock Exclusion},
  journal = {Review of Finance},
  year    = {2022},
  volume  = {26},
  number  = {6},
  pages   = {1345--1388},
  doi     = {10.1093/rof/rfac045}
}

@article{WangWangDongWang2024IRFA,
  author  = {Wang, Jianli and Wang, Shaolin and Dong, Minghua and Wang, Hongxia},
  title   = {{ESG} Rating Disagreement and Stock Returns: Evidence from {China}},
  journal = {International Review of Financial Analysis},
  year    = {2024},
  volume  = {91},
  pages   = {103043},
  doi     = {10.1016/j.irfa.2023.103043}
}

@article{HuLiLi2024IRFA,
  author  = {Hu, Kexin and Li, Xingyi and Li, Zhongfei},
  title   = {Effect of {ESG} Rating Disagreement on Stock Price Informativeness: Empirical Evidence from {China's} Capital Market},
  journal = {International Review of Financial Analysis},
  year    = {2024},
  volume  = {96},
  number  = {PB},
  pages   = {103651},
  doi     = {10.1016/j.irfa.2024.103651}
}

@article{HePanShanZhou2025IRFA,
  author  = {He, Ye and Pan, Yuetong and Shan, Tao and Zhou, Yanyu},
  title   = {{ESG} Rating Disagreement and the Cost of Equity Financing},
  journal = {International Review of Financial Analysis},
  year    = {2025},
  volume  = {107},
  pages   = {104565},
  doi     = {10.1016/j.irfa.2025.104565}
}

@misc{SCMalaysiaLists2013_2025,
  author       = {{Securities Commission Malaysia}},
  title        = {{[dataset] List of Shariah-Compliant Securities, Semi-Annual Releases, 2013--2025}},
  year         = {2025},
  howpublished = {Securities Commission Malaysia Islamic Capital Market Publications},
  url          = {https://www.sc.com.my/development/icm/icm-publications/list-of-shariah-compliant-securities},
  note         = {Accessed April 29, 2026}
}
